\documentclass[format=acmsmall, review=false]{acmart}
\usepackage{acm-ec-26}
\usepackage{booktabs} % For formal tables
\usepackage[ruled]{algorithm2e} % For algorithms

\usepackage{dsfont}
\usepackage{mathtools}
\usepackage{amsmath}
\usepackage{acronym}
\usepackage{cleveref}
\newacro{iid}[i.i.d.]{Independent and Identically Distributed}
\newacro{bdp}[BDP]{Backward Dynamic Programming}
\newacro{ldp}[LDP]{Local Differential Privacy}

\newcommand{\hr}[1]{\textcolor{blue}{[Hugo: #1]}}

\newcommand{\mathieu}[1]{\textcolor{teal}{[Mathieu: #1]}}

\newcommand{\ind}[1]{\mathds{1}\left\{#1\right\}}
\newcommand{\dham}{d_{\mathrm{Ham}}}
\newcommand{\Dcal}{\mathcal{D}}
\newcommand{\NN}{\mathbb{N}}

\newcommand{\EE}{\mathbb{E}}
\newcommand{\mech}{\mathcal{Q}}
\newcommand{\RR}{\mathbb{R}}
\newcommand{\PP}{\mathbb{P}}
\newcommand{\Vn}{\mathcal{V}}
\newcommand{\Mn}{\mathcal{M}}
\newcommand{\Wn}{\mathcal{W}_\varepsilon}
\newcommand{\Ln}{\mathcal{L}_\varepsilon}

\DeclareMathOperator*{\argmax}{arg\,max}

\newtheorem{remark}{Remark}

\newtheorem{assumption}{Assumption}

\SetAlFnt{\small}
\SetAlCapFnt{\small}
\SetAlCapNameFnt{\small}
\SetAlCapHSkip{0pt}
\IncMargin{-\parindent}

\usepackage{nicefrac}
\setcitestyle{authoryear}

\title[Prophet Inequalities under LDP]{Prophet Inequalities under Local Differential Privacy}

\author{Achraf Azize\textsuperscript{*}}
\affiliation{%
  \institution{CREST, ENSAE, FairPlay Joint Team, IP Paris}
  \country{France}
}

\author{Mathieu Molina\textsuperscript{*}}
\affiliation{%
  \institution{Tel Aviv University}
  \country{Israel}
}

\author{Hugo Richard\textsuperscript{*}}
\affiliation{%
  \institution{Criteo AI Lab, FairPlay Joint Team}
  \country{France}
}

\author{Vianney Perchet}
\affiliation{%
  \institution{CREST, ENSAE, Criteo AI Lab, FairPlay Joint Team}
  \country{France}
}

\renewcommand{\shortauthors}{Azize, Molina, Richard, and Perchet}

\begin{abstract}

Many online decision platforms, from hiring marketplaces to auctions, face a tension between efficient decision-making and the protection of participants' privacy. Personal information, such as a candidate’s test score or a bidder’s valuation linked to protected data, is sensitive, and fear of data resale or reputational harm can make participants reluctant to share it. Furthermore, platforms can be untrusted or even incentivized to resell data, making local privacy guarantees that do not rely on a trusted centralized curator preferable. We initiate the study of optimal stopping and prophet inequalities under local differential privacy (LDP). Each of $n$ independent arriving values is observed only through reports generated by an $\varepsilon$-LDP mechanism. The decision maker must design the $\varepsilon$-LDP mechanisms, and choose an irrevocable stopping time to maximise the expected selected \emph{true} value. We characterize the optimal online stopping rule under LDP and show that simple binary mechanisms suffice: an optimal LDP stopping rule can be implemented via a randomized-response-type report and a dynamic-programming threshold rule. We then quantify performance via tight competitive ratios against two benchmarks. Relative to the \emph{optimal non-private online} policy, we prove a tight worst-case competitive ratio of $e^{\varepsilon}/(n - 1 + e^{\varepsilon})$, interpolating between $1/n$ (full privacy) and $1$ (no privacy). Relative to an \emph{LDP prophet}, who designs $\varepsilon$-LDP mechanisms but observes the full privatized sequence before deciding what to select, we prove a tight competitive ratio of $(1 + e^{-\varepsilon})/2$, interpolating between $1$ (full privacy) and the classical $1/2$ bound (no privacy). Notably, increasing privacy shrinks the LDP prophet's advantage faster than it degrades online performance, closing the performance gap.

\end{abstract}

\begin{document}

% Title page for title and abstract only.
% \begin{titlepage}

\maketitle
\begingroup
\renewcommand\thefootnote{*}
\footnotetext{Equal contribution.}
\endgroup

% Optionally include a table of contents
\vspace{1cm}
\setcounter{tocdepth}{1} % adjust to 1 if desired
\tableofcontents

% \end{titlepage}

% Paper body
\section{Introduction}

% Optimal stopping intro
Online stopping problems lie at the core of decision-making under uncertainty. In their simplest form, they model situations in which a decision maker observes a sequence of values $X_1, X_2, \dots$ and must decide, at each step in time $t \in\{1, 2, \dots \}$, whether to stop and irrevocably select the current option gaining $X_t$ or to continue in the hope of encountering a better one in the future. This framework has a long and influential history in economics, probability, and operations research, and underpins a wide range of applications, from job search and hiring to asset sales, procurement, and dynamic allocation problems.

% Hiring / ads / markets example
A canonical example of online stopping is hiring. A firm interviews candidates sequentially, each with an unknown quality, and must decide immediately whether to make an offer or move on. Once a candidate is rejected, they are no longer available. The firm’s objective is to hire a strong candidate despite not knowing the quality of future candidates and being unable to recall past applicants. Similar stopping decisions arise in online advertising: an advertiser observes a stream of user visits on a website, each associated with an uncertain instantaneous value (e.g., expected revenue or engagement from showing a premium ad), and must decide whether to buy the ad-slot now or wait for potentially higher-value users later. Similarly, in online markets and allocation problems, a seller faces buyers arriving over time and must decide whether to allocate a limited resource (a single item, a scarce service slot, or a budgeted promotion) to the current buyer or to keep it for future arrivals. Optimal stopping theory provides principled ways to reason about such problems, and prophet inequalities, in particular, offer a powerful performance benchmark: they quantify how well an online decision-maker can perform relative to a hypothetical ``prophet" who observes all past and future values.

% Privacy concerns in hiring / ads/ markets
Modern decision-making increasingly takes place in environments where information is not only valuable but also sensitive. In the hiring example, candidate quality is assessed through signals such as test scores, interviews, recommendation letters, or salary expectations. These signals may reveal personal characteristics or strategic information that candidates are reluctant to disclose fully. Candidates may worry that detailed evaluations could be reused for future hiring decisions, shared across subsidiaries, or otherwise exploited in ways that affect their bargaining power or career prospects. 
Similar concerns arise in online advertising, where a user’s ``value'' is inferred from rich, fine-grained signals (browsing history, purchases, location, device fingerprints, inferred interests), and where revealing or storing such information can enable cross-site tracking, sensitive attribute inference, or discriminatory targeting. Users may therefore be unwilling to share accurate signals if they anticipate profiling, resale to third parties, or long-term linkage across services, even when doing so would improve relevance or efficiency.
Auctions and market platforms raise privacy concerns that closely parallel those in hiring.
A buyer’s valuation is highly sensitive information: revealing it can expose willingness-to-pay, affect future personalised pricing, or weaken bargaining positions in other transactions.
As emphasised in recent work on privacy in auctions~\citep{Eilat2025}, even limited disclosure of valuation information can have significant strategic and welfare consequences.
Buyers may therefore be unwilling to report exact values to a platform, even when doing so would improve allocation efficiency.
More broadly, regulatory constraints and data protection norms increasingly limit the extent to which platforms or firms can collect, store, and process fine-grained personal data. This introduces a fundamental friction that is absent from classical stopping models: even when candidates are willing to participate in the hiring process, they may not trust the decision-maker with their exact private information. The act of revealing one’s value or quality becomes costly in itself.

To understand the tension between privacy and utility in online stopping problems, we focus on characterising how much performance must be sacrificed to protect the privacy of individual agents. We measure utility using the standard competitive-ratio framework from prophet inequalities, i.e. comparing the expected outcome of an online stopping rule to appropriate benchmarks. Privacy is formalised using Differential Privacy (DP)~\citep{Dwork_Calibration}, which has emerged as the gold standard for rigorous privacy guarantees in data analysis.

%DP ensures the protection of an individual’s sensitive information when their data is used for analysis. Specifically, 
DP is a worst-case constraint on the class of randomised mechanisms: a mechanism satisfies the DP constraint if the mechanism’s outputs are ``essentially” equally likely to occur, for any two input datasets that only differ in one individual’s data. 
% \hr{we should focus on Local DP} \achraf{My goal here is sell DP as a good framework to define privacy for these type of problems in general. LDP, for me, is more of a ``trust'' model inside DP. } 
By designing a mechanism that satisfies the DP constraint, the decision-maker honours the following ``privacy promise": whatever would have happened to any user due to their participation in the decision process, would likely have happened if they did not participate. The worst-case nature of DP ensures that this promise is honoured, even if the user has a very unlikely data point (e.g. an outlier), and for any coalition of the other users. The effectiveness of the DP definition lies in its information-theoretic nature, in the sense that DP protects against any adversary, with unlimited amounts of computational power and auxiliary information, trying to infer private information about users by observing the output of the mechanism. 

Our primary focus in this paper is on the Local Differential Privacy (LDP) trust model: under LDP, agents randomise their values before they are observed by the decision-maker. The platform never sees raw values, even temporarily. This model is particularly compelling in environments where trust in centralised data collection is limited, or where regulatory and institutional constraints require strict privacy protection. Importantly, LDP imposes stronger restrictions than standard (global) DP and therefore presents a more stringent test of what can be achieved under privacy constraints. All of our motivating examples fit naturally into the LDP framework.
In hiring, candidates can privately randomize their evaluation scores before submitting them to the firm, so that the firm never observes the exact score.
In online advertising, users (or their devices) can locally privatize signals derived from sensitive features before they are used for targeting or allocation, so that the platform only sees a noisy report rather than raw browsing or purchase history.
In auctions and online markets, bidders can locally privatize signals about their valuations before interacting with the seller or platform, preventing precise willingness-to-pay from being revealed.
In all these settings, the decision maker must act online and irrevocably based solely on these privatised reports.
Motivated by these examples, we study a unified problem: \emph{online optimal stopping under local differential privacy (LDP)}. There are $n$ sequentially arriving agents, where agent $i \in \{1, \dots, n\}$ holds a private value $X_i$ drawn from a known distribution. The decision maker must choose an irrevocable stopping time, and hence which agent (if any) is selected, but, crucially, never gets to observe the raw values. Instead, each agent sends a \emph{locally privatized} report $Z_i$ produced by an $\varepsilon$-LDP mechanism applied to $X_i$. Thus, for any two possible values $x,x'$ of an agent, the distribution of the reported message is nearly indistinguishable:
no single agent’s report can strongly reveal whether their value was $x$ or $x'$.
The decision maker observes the sequence of reports $(Z_1,\dots,Z_n)$ online and must decide when to stop based solely on these noisy messages. The algorithmic objective is to design both a stopping rule and the LDP reporting scheme that respects the $\varepsilon$-LDP constraint for every agent while achieving strong performance guarantees in the prophet-inequality sense.
In particular, we seek to characterise the optimal privacy-utility trade-off: how much competitive performance is lost when the decision maker is restricted to act on locally privatized information. We evaluate performance through competitive ratios against three natural benchmarks:
(a) the optimal non-private online stopping value,
(b) the standard non-private prophet value,
and (c) a \emph{private prophet} benchmark that is also constrained to observe only $\varepsilon$-LDP messages.

% An online decision-maker observes sequentially some values, where each value is associated with an agent's private information, and must make an irrevocable ``differentially private" stopping decision, i.e. the stopping decision itself should ``essentially” equally likely to occur, for any two sequences of input values that only differ in one agent’s value. The objective is to design stopping rules that ensure differential privacy while achieving optimal competitive guarantees.

% Specify that we are interested in local DP, and that both our examples from before (hiring, auctions) fall in the local DP setting

% Both of our motivating examples fit naturally into the local DP framework. In hiring, candidates can privately randomise their evaluation scores before submitting them to the firm. In auctions, bidders can locally privatise signals about their valuations before interacting with the seller. In both cases, the platform must make online, irreversible decisions based solely on these privatised ``noisy" reports. 

This leads to our central questions:

\begin{center}
    \textit{What is the optimal stopping algorithm under local differential privacy?\\
How does its performance compare to that of a non-private online decision-maker, a non-private prophet, and a privacy-preserving prophet?}
\end{center}

\paragraph{Contributions} Answering these questions led to the following contributions.

1. \textbf{Optimal stopping algorithm under LDP.} Our first contribution is an algorithmic characterization of an \emph{optimal} $\varepsilon$-LDP stopping strategy, jointly over the reporting mechanisms and the stopping rule.
We show that the infinite-dimensional design space of $\varepsilon$-LDP channels can be reduced \emph{without loss} to a two-message family: an optimal policy exists in which, at each time step, the arriving agent releases a \emph{binary} ($1$-bit) privatized report.
This yields a sharp simplification of the problem and leads to an explicit backward-induction characterization:
the optimal policy is described by a dynamic program over continuation values, where each stage consists of (i) choosing an optimal binary $\varepsilon$-LDP mechanism for that stage and (ii) applying the corresponding optimal acceptance rule on the received bit.
In particular, optimality can be implemented by a clean ``privatize-then-decide'' procedure with time-dependent thresholds determined by the continuation value recursion.
Beyond providing an implementable description, this result pinpoints precisely \emph{how} the privacy constraint limits the information revealed by each arrival and how the optimal policy exploits the remaining information.

2. \textbf{Competitive ratios against non-private benchmarks.} Our second contribution quantifies the \emph{price of LDP} by comparing the optimal $\varepsilon$-LDP stopping value $\Wn$ to two \emph{non-private} benchmarks: (a) We show a tight competitive ratio against the optimal non-private online stopping value $\Vn$. LDP induces a worst-case degradation  of $e^{\varepsilon} / (n-1 + e^{\varepsilon})$ , compared to $\Vn$. This formula makes the dependence on both the horizon $n$ and the privacy level $\varepsilon$ explicit,
interpolating between $1/n$ in the full-privacy regime and $1$ in the no-privacy regime.
(b) We then compare $\Wn$ to the classical prophet benchmark $\Mn=\mathbb{E}[\max_{i\in[n]}X_i]$.
Here we provide matching-order bounds: we prove an upper bound and a constructive lower bound on the worst-case ratio, pinning down its dependence on $(n,\varepsilon)$ up to a constant factor (\Cref{prop:wagainstm}).
In addition, beyond the optimal dynamic policy, we exhibit a substantially simpler mechanism-policy pair: a \emph{single-threshold} stopping rule coupled with a time-homogeneous Randomised Response $\varepsilon$-LDP binary reporting scheme. We show that this simple policy already achieves the tight worst-case guarantee against the non-private online benchmark,
and simultaneously attains the lower bound in our comparison to the non-private prophet
(\Cref{prop: private T / M}).
This provides a direct private analogue of classical threshold rules, and demonstrates that extremely simple stopping strategies can achieve strong privacy-utility trade-offs. 

3. \textbf{Competitive ratios against private benchmarks.} Our third contribution studies performance relative to \emph{privacy-preserving} offline benchmarks, where the prophet is itself constrained by $\varepsilon$-LDP.
Let $\Ln$ denote the value of an optimal \emph{private prophet} offline strategy under $\varepsilon$-LDP.
We show that the worst-case competitive ratio of $\Wn$ against this private prophet $\Ln$ admits a \emph{tight} closed form $(1 + e^{-\epsilon})/2$, recovering the classical $1/2$ bound as $\varepsilon\to\infty$ and approaching $1$ in the full-privacy regime
(\Cref{thm:w/l}).
This characterizes the \emph{price of adaptivity under LDP}: once both online and offline decision makers are restricted to locally privatized information, the prophet’s advantage shrinks exactly to $(1+e^{-\varepsilon})/2$. Finally, we strengthen this benchmark comparison by showing that the private prophet does not benefit from richer message alphabets:
even if the offline private benchmark is additionally restricted to use \emph{binary} ($1$-bit) LDP reports, its value remains sufficient to attain the same worst-case ratio. This result mirrors our structural reduction on the online side and shows that, for worst-case competitive analysis, \emph{one-bit local privacy is essentially without loss} for both the optimal online rule and the private offline benchmark.

Together, these results provide a sharp understanding of the trade-offs between privacy and performance in prophets inequalities.

\begin{table}[t!]
\centering
\caption{\textbf{Bounds on various competitive ratios.} $C_n$ is the set of non-negative random variables $X_1, \dots, X_n$ such that $X_1, \dots, X_n$ are independent. $\Vn$ is the value of the optimal non-private online algorithm, $\Mn$ is the value of the optimal non-private offline algorithm, $\Wn$ is the value of the optimal $\varepsilon$-LDP online algorithm, $\Ln$ is the value of the optimal $\varepsilon$-LDP offline algorithm. Novel results are in \textcolor{blue}{blue}.}
\begin{tabular}{lccc}
\toprule
\textbf{Competitive ratio} & \textbf{Lower bound} & \textbf{Upper bound} & \textbf{Reference} \\
\midrule
$\inf_{\mathbf{X} \in C_n} \Vn(\mathbf{X}) / \Mn(\mathbf{X})$ & $1 / 2$ & $1 / 2$ & \citeauthor{KrengelSucheston1977Semiamarts} \\[4pt]
$\inf_{\mathbf{X} \in C_n} \Wn(\mathbf{X}) /\Vn(\mathbf{X})$ & \textcolor{blue}{$e^{\varepsilon} / (e^{\varepsilon} - 1 + n)$} & \textcolor{blue}{$e^{\varepsilon} / (e^{\varepsilon} - 1 + n)$} & \Cref{thm:cr_stopping_ldp} \\[4pt]
$\inf_{\mathbf{X} \in C_n} \Wn(\mathbf{X}) /\Mn(\mathbf{X}) $ & \textcolor{blue}{$e^{\varepsilon} / (2e^{\varepsilon} - 2 + 2n) \vee 1/n $} & \textcolor{blue}{$e^{\varepsilon}/(2e^{\varepsilon} - 2 + n)$} & \Cref{prop:wagainstm} \\[4pt]
$\inf_{\mathbf{X} \in C_n} \Wn(\mathbf{X}) /\Ln(\mathbf{X})$ & \textcolor{blue}{$(1 + e^{-\varepsilon}) / 2 $} & \textcolor{blue}{$(1 + e^{-\varepsilon}) / 2 $} & \Cref{thm:w/l} \\
\bottomrule
\end{tabular}
\label{tab:bornes}
\end{table}

\section{Related Work}

Classical prophet inequalities compare the optimal online stopping value to the offline maximum, with a tight worst-case ratio $1/2$ for independent nonnegative variables \citep{KrengelSucheston1977Semiamarts,SamuelCahn1984,hill1992survey}.
Several works subsequently explored richer information models.
Most closely related, \citet{Assaf1998ASV} study stopping with partial observations, where the decision-maker observes a signal correlated with the underlying value, and show that when the joint model is known, prophet-type guarantees (including a $1/2$ bound) can still be achieved.
The model we study can be viewed as a \emph{design} version of partial-observation stopping: the observation channel is not fixed exogenously, but is induced by an $\varepsilon$-LDP mechanism chosen by the decision-maker, and must be optimized jointly with the stopping rule.
A complementary line of work assumes that distributions are not known a priori but are accessed through a small number of accurate samples \citep{Azar2014limited}; in particular, \citet{Rubinstein2019OptimalSP} show that a single sample per distribution suffices to recover the $1/2$ competitive ratio.
This is, in a sense, dual to the privacy setting studied here: we assume full prior knowledge of the distributions, but each realized observation is systematically corrupted because of the privacy constraints.

Differential privacy (DP) \citep{dwork2014algorithmic} is a framework that provides rigorous privacy guarantees for randomized algorithms; in the central model, a curator has access to raw data and must ensure that published outputs do not reveal any individual's information. When such trust is not appropriate,
%, as is the case we are interested in, 
local Differential Privacy (LDP) requires each agent to privatize their value before it is observed, often via randomized response \citep{warner1965randomized,dinur2003revealing}.
A central theme in the LDP literature is to characterize and compute optimal mechanisms under various utility criteria: \citet{Duchi2016MinimaxOP} focus on statistical tasks such as mean estimation, while \citet{kairouz2014extremal} provides structural results for sub-linear utility objectives (including testing losses).
The problem we study differs from canonical estimation or testing objectives: utility is induced by an online stopping objective, and the mechanism is chosen to maximize a sequential value criterion under known priors.
Conceptually, locally private optimal stopping is also related to DP primitives for selecting large elements (e.g., noisy-max and sparse-vector-type ideas in central DP \citep{dwork2014algorithmic,Lyu2016UnderstandingTS} or the localisation phase of mean and histogram estimation algorithms in user-level LDP \cite{kent2024rate,acharya2023discrete,pla25}). However, in locally private optimal stopping, the mechanism is optimized for a Bayesian stopping objective, rather than a pure frequentist worst-case approach.

Privacy constraints have been studied extensively in online decision-making, with emphasis on characterizing privacy-utility trade-offs typically expressed in regret, sample complexity, or success probability. In multi-armed bandits, a number of works derive near-optimal privacy-dependent regret bounds and propose private algorithms under central DP~\citep{dpseOrSheffet, azize2022privacy, azizeoptimal} and LDP~\citep{localDP, ren2020multi}. Privacy has also been investigated in pure exploration problems closely related to stopping, such as best-arm identification, with nearly optimal fixed-confidence algorithms under central DP~\citep{azize2023complexity, jourdan2025optimal}, as well as for LDP variants \citep{azize2024differentially}.
Closer to prophet inequalities models, \citet{Grining2020WhatDO} study the secretary problem under central DP and quantify the impact of privacy on selection guarantees.
Compared to these works, our focus is on LDP, motivated by settings where raw values are never entrusted to the platform due, for instance, to concerns of data resale, and on the Bayesian prophet model, where prior information is available and must be used to design optimal LDP mechanisms and stopping rules.

Prophet inequalities are naturally related to market design: they have been used to analyze simple sequential posted-pricing mechanisms~\citep{Hajiaghayi2007AutomatedOM,Chawla2010multiparam,Lucier2017AnEV}. A line of equivalence results relates approximation guarantees for revenue-maximizing auctions via posted prices to competitive analysis in prophet inequalities \citep{Feldman2014CombinatorialAV,CORREA201925}.
In auctions, privacy is a natural concern since valuations are sensitive and can be strategically exploited if revealed.
A first line of work focuses on cryptographic methods for privacy-preserving auctions \citet{Naor1999Privacyauction,Brandt2005privacymultiunit}, aiming to reveal only what is implied by the outcome, that is to say hiding losing bids; see~\citet{ALVAREZ2020101502} for a survey.
A second, more recent direction, 
incorporates \emph{statistical} privacy notions directly into auction formats, including Differential Privacy in call auctions~\citep{Diana2020callauction}, double auctions with DP guarantees~\citep{JIANG202213}, and double auctions under local DP~\citep{Huang2025doubleauction}. Complementarily, \citet{Eilat2025} study a Bayesian notion of privacy based on mutual information and analyze which auction formats minimize information leakage under additional assumptions.
Our work is aligned with the Differential-Privacy-based line of research, but addresses a different primitive than the auction formats above: we study the underlying \emph{prophet} stopping problem under \emph{local DP}.
In our model, privacy is enforced at the observation level (agents locally privatize their realized values), and the decision maker must act online on privatized reports.

\section{Model}
We first provide the essential background in local differential privacy and optimal stopping in~\Cref{subsec:privacybackground} and \Cref{subsec:stoppingbackground} respectively. Then, in~\Cref{subsec:online stopping}, we formally describe the problem of \emph{online optimal stopping under LDP}, our main object of study.

\paragraph{Notations} For $i < j$, we denote $a_{i:j}$ the sequence $(a_k)_{k \in [i, j]}$. If $X$ is a random variable of distribution $F$ and $Y$ is a random variable such that $Y \mid X=x$ is given by the transition kernel $G(\cdot \mid x)$, the joint distribution of $(Y, X)$ is $G \otimes F$. Given $\mathbf{X} = X_{1:n}$ a sequence of independent random variable with joint distribution $F_1 \otimes \dots \otimes F_n$, the notations $\Mn(\mathbf{X}), \Ln(\mathbf{X}), \Vn(\mathbf{X}), \Wn(\mathbf{X})$ represent, respectively, the expected value of the non-private prophet, private prophet, non-private optimal online algorithm and private optimal online algorithm on the instance $X_{1:n}$. It is an abuse of notation as these quantities are deterministic and only depend on the \emph{distribution} of $X_{1:n}$.

\subsection{Local Differential Privacy}\label{subsec:privacybackground}

Differential Privacy~\cite{Dwork_Calibration} is a constraint on the class of randomised mechanisms. In general, a randomised mechanism $M$ takes as input a dataset $D$, which is a collection of $n$ users data points $D \coloneqq \{x_1, \dots, x_n\} \in \mathcal{X}^n$, from an input universe $\mathcal{X}$. The mechanism $M$ outputs a distribution $M_D \in \mathcal{P}(\mathcal{O})$, where $\mathcal{P}(\mathcal{O})$ is the set of distributions over the output space $\mathcal{O}$. The probability space is over the coin flips of the mechanism $M$. Given some measurable event $E$ in  $\mathcal{O}$, we note $M_D(E) \coloneqq M(E | D)$ the probability of observing the event $E$ given that the input of the mechanism is $D$. Given two datasets $D \coloneqq \{ x_1, \dots, x_n\}$ and $D' \coloneqq \{ x'_1, \dots, x'_n\}$ in $\mathcal{X}^n$, let $\dham(D, D') \coloneqq \sum_{i=1}^n \ind{D_i \neq D'_i}$ denote the Hamming distance between $D$ and $D'$, i.e. the number of different data points between $D$ and $D'$. We say that $D$ and $D'$ are neighbouring datasets, that we note $D \sim D'$, if and only if $\dham(D, D') \leq 1$, i.e. $D$ and $D'$ differ by at most one record.

Now, we formally define Differential Privacy (DP). A mechanism $M$ satisfies $\varepsilon$-DP~\cite{Dwork_Calibration} for a given $\varepsilon \geq 0$ if,
\begin{align}\label{eq:DP}
    \forall D \sim D', \, \forall E \subset \mathcal{O}\footnote{When we write $E \subset \mathcal{O}$, we mean $E$ is a measurable subset of $\mathcal{O}$.}, \, M(E | D) \leq e^\varepsilon M(E | D')
\end{align}

A mechanism $M$ satisfies DP if the mechanism behaves "similarly" on \emph{all} neighbouring datasets $D$ and $D'$, and on all realisations of the mechanism $M$. By designing a mechanism $M$ that satisfies the DP constraint, the data analyst honours the privacy "promise": whatever would have happened to any user due to their participation in a DP analysis, i.e. the world where $o \sim M_D$, would likely have happened if they did not participate, i.e. the world where $o \sim M_{D'}$, for $D' \sim D$. 
The worst-case nature of DP ensures that the promise is honoured, even if the user has a very unlikely data point (e.g. an outlier), and for any coalition of the other users. 
The effectiveness of the DP definition lies in its information-theoretic nature: DP protects against any adversary with unlimited amounts of computational power and auxiliary information.

The adversary can formulate a privacy attack as a binary hypothesis test. Specifically, based only on the observed output $o$ of the mechanism $M$, the adversary needs to determine whether $H_0$: ``The output was generated from the dataset $D$", vs $H_1$: The output was generated from a neighbouring dataset $D'$. DP imposes a trade-off between the Type I and Type II errors of any adversary trying to conduct this hypothesis test~\cite{kairouz2015composition}. Specifically, it is impossible for any adversary to get both Type I and Type II errors to be small when the mechanism verifies DP. This also provides an operational interpretation of the DP constraint.

In this paper, we study the setting of \emph{local} Differential Privacy, where users do not trust the data analyst, i.e. the entity collecting the data and running the mechanism, with their raw data. Thus, in this setting, a mechanism $Q$ is applied locally to each granular data point $x \in \mathcal{X}$. The data analyst only has access to the noisy data points $z \sim Q(\cdot | x)$. Local DP is one of the oldest formulations of privacy, dating back to~\citet{warner1965randomized}, who advocated it as a solution to what he called ``evasive answer bias" in survey sampling. A mechanism $Q$ satisfies $\varepsilon$-local DP if, seen as a mechanism applied to a dataset of size $1$, it verifies $\varepsilon$-DP. Specifically,  A randomised mechanism $Q$ satisfies $\varepsilon$- LDP~\citep{dinur2003revealing, evfimievski2003limiting} if
\begin{equation}\label{eq:ldp}
     \forall x, x' \in \mathcal{X}, \,
     \forall S \subset \mathrm{Range}(Q), \, Q(S \mid x)  \leq e^{\varepsilon} Q(S \mid x'),
\end{equation}
where the probability space is over the coin flips of the mechanism $Q$. We denote $\mathcal{Q}_\varepsilon$ the set of $\varepsilon$-LDP mechanisms. Now, under LDP, and given a dataset $D=\{x_1,\dots,x_n\}$, the analyst observes only the privatized dataset $\tilde D=\{z_1,\dots,z_n\}$, where $z_i \sim Q(. \mid x_i)$, and may apply \emph{any} (possibly randomized) function $f$ to $\tilde D$, without weakening privacy: by the post-processing property of DP, the composed mechanism $D \mapsto \tilde D \mapsto f(\tilde D)$ remains $\varepsilon$-DP.

For binary attributes, the Randomised Response (RR) mechanism~\citep{warner1965randomized} is a popular way to achieve $\varepsilon$-local DP. The idea is to output the true value of a user's response with probability $e^\varepsilon/(e^\varepsilon+1)$ and output the opposite value with probability $1/(e^\varepsilon + 1)$.
To make it suitable for larger discrete domains, a Generalised Randomised Response (GRR) is proposed in~\citet{kairouz2016discrete}. 
For continuous numerical data statistics, adding Laplace noise to each data record achieves local DP~\cite{dwork2014algorithmic}. Although local DP is a relatively stringent requirement, we view this
setting as a natural first step in understanding optimal stopping problem under privacy constraints. 

\subsection{Bayesian Optimal Stopping and Prophet Inequalities}\label{subsec:stoppingbackground}

Here, we briefly recall some results on the standard optimal stopping problem, as well as the classic prophet inequalities result. They relate the performance of the optimal stopping value to that of the omniscient decision-maker. While this material is standard, we recall it in detail here in order to best highlight the difference when privacy is introduced as an additional requirement.

Let $n \in \mathbb{N}^*$, $\mathcal{C}_n$ be the set of $n$ positive independent random variables with finite expectation and consider $\mathbf{X} = X_{1:n} \in \mathcal{C}_n$ with joint distribution $F_1 \otimes \dots \otimes F_n$. At time $i \in [n]$, the decision maker observes $X_i$ and must decide whether to stop and gain $X_i$ or continue. The goal of an optimal stopping problem is to find a stopping rule that maximizes the expected gains. Formally, we define by $\mathcal{F}=(\mathcal{F}_i)_{i \in [n]}$ the natural filtration of the $X_i$, that is to say $\mathcal{F}_i=\sigma(X_{1:i})$. We denote by $\mathcal{T}^{(n)}$ the set of stopping rules taking values in $[n] \cup \{ \infty \}$ adapted to $\mathcal{F}$, so that $\tau \in \mathcal{T}^{(n)}$ means that for $i \in [n]$ the event $\{ \tau = i\}$ is $\mathcal{F}_i$ measurable. For a given stopping rule $\tau \in \mathcal{T}^{(n)}$, the realized value at stopping time is $X_{\tau} \coloneqq \sum_{ i \in [n]} X_i \mathds{1}[\tau=i]$, hence if $\tau=i$ then $X_{\tau}=X_i$, with the convention that $X_{\infty}=0$. The optimal strategy for the decision maker is given by the stopping rule $\tau^* \in \mathcal{T}^{(n)}$ that reaches the maximum expected value $\Vn(\mathbf{X})$ at stopping time,
\begin{equation}
\Vn(\mathbf{X})=\max_{\tau \in \mathcal{T}^{(n)}} \mathbb{E}[X_{\tau}].
\end{equation}
\citet{Chow1971GreatET} shows that there exists an optimal stopping rule that can be obtained via a sequence of deterministic thresholds computed through backward induction. More precisely, it holds that $\Vn(\mathbf{X}) = V_1(\mathbf{X})$ where $(V_i(\mathbf{X}))_{i \in [n+1]}$, the \ac{bdp} sequence, is defined as $V_{n+1}(\mathbf{X})=0$ and for any $i \in [n]$,
\begin{align}
\label{eq:bdp}
V_i(\mathbf{X})&=\mathbb{E}\left[\max\left(X_i,V_{i+1}(\mathbf{X}) \right)\right].
\end{align}
The threshold $V_i(\mathbf{X})$ represents the expected value that can be obtained starting from time $i$. The optimal stopping rule $\tau^*$ is then given by 
 \[
 \tau^* = \min \{i \in [n], X_i \geq V_{i+1}(\mathbf{X}) \}.
 \] 

In prophet inequalities, the goal is to compare this optimal online stopping value $\Vn(\mathbf{X})$, with the performance that can be achieved by an omniscient prophet having observed all the values. If the prophet knows all the values in advance, then the best it can do is select the maximum when it arrives, achieving the performance of
\begin{equation}
\Mn(\mathbf{X})=\mathbb{E}\left[\max_{i \in [n]} \left( X_i \right)\right].
\end{equation}
One of the main questions in prophet inequalities, is to compare $\Vn$ and $\Mn$ in the worst case. \citet{KrengelSucheston1977Semiamarts} show that the worst case ratio is exactly $1/2$, i.e. for all $n \in \mathbb{N}^*$,
\begin{equation}\label{eq:class_proph}
\inf_{\mathbf{X} \in \mathcal{C}_n} \frac{\Vn(\mathbf{X})}{\Mn(\mathbf{X})}= \frac{1}{2}.
\end{equation}
% \achraf{Let's mention in this equation that $X_1 \dots X_n$ are $\geq 0$ and ind. like in the other CRs we present later. Also, maybe let us just call it infinimumm over $\mathbf{X} \in \mathcal{C}_{\geq 0}^{ \text{ind.}}?$}
In this paper, we investigate the relative value, compared to different private and non-private benchmarks, of an online decision maker who must satisfy the LDP constraint. 

\subsection{Online Stopping under Local Differential Privacy}
\label{subsec:online stopping}
We now introduce the core aspect of online stopping under \ac{ldp}. Let $n \in \mathbb{N}^*$ and $\mathbf{X} \in \mathcal{C}_n$  with joint distribution $F_1 \otimes \cdots \otimes F_n$. In addition, fix $\varepsilon \geq 0$, a privacy parameter. At time $i \in [n]$, the online decision maker does not observe $X_i$ but instead observes a noisy version $Z_i$ of $X_i$ generated via a $\varepsilon$-LDP mechanism $Q_i$ and must decide whether to stop and gain $X_i$ or continue. The goal of an optimal stopping problem under \ac{ldp} is to find a set of mechanisms $Q_{1:n}$ and a stopping rule to maximize the expected $\emph{true}$ gains. 

Formally, given $Q_{1:n} \in \mathcal{Q}_\varepsilon^n$, for any $x \in \mathcal{X}$ and any $i \in [n]$,
\[
(Z_i  \mid X_i=x) \sim Q_i( \cdot \mid x).
\]

\begin{comment}
Consider $Q_1, \dots, Q_n \in \mathcal{Q}_\varepsilon$ a set of  We call $\mathcal{Q}_\varepsilon$ and distribution. $\mathcal{Z}_i$ some (to be chosen) message output space. Let $\mathcal{P}(\mathcal{Z}_i)$ be the set of probability distributions over $\mathcal{Z}_i$, we define the set of $\varepsilon$-\ac{ldp} mechanisms  (or $\varepsilon$-\ac{ldp} Markov kernels) as
\begin{equation}\label{eq:ldp_mechs}
\mathcal{Q}_{i, \varepsilon}\coloneqq \left\{ Q_i : \mathcal{X}_i \rightarrow \mathcal{P}(\mathcal{Z}_i), \text{ such that }\forall S \subset \mathcal{Z}_i,\ \forall x,x'\in\mathcal{X}_i,
Q_i(S|x)\le e^{\varepsilon}\,Q_i(S|x') \right\}.
\end{equation}
In the setting studied in this paper, the online decision maker is unable to observe the $X_i$ directly. Instead, she observes some messages $Z_1,\dots,Z_n$ in an online fashion where 
\begin{equation}
Z_i \sim Q_i( \, \cdot \, \mid X_i).
\end{equation}
\end{comment}

Given that the $X_i$ are independent, so are the $Z_i$. The observable information therefore corresponds to the filtration $\mathcal{O}=(\mathcal{O}_i)_{i \in [n]}$ where $\mathcal{O}_i=\sigma(Z_{1:i})$. We denote $\mathcal{T}(Q_{1:n})$ the set of stopping rules adapted to $\mathcal{O}$. The goal of an optimal stopping problem under LDP is then to find the set of mechanisms $Q_{1:n}^* \in \mathcal{Q}_\varepsilon^n$ and the stopping rule $\tau^* \in \mathcal{T}(Q_{1:n})$ that reach the maximum expected value $\Wn(\mathbf{X})$ at stopping time, for the \emph{original data} $\mathbf{X}$,
\begin{equation}
\label{eq:Wn}
\Wn(\mathbf{X}) \coloneqq \max_{Q_{1:n} \in \mathcal{Q}_\varepsilon^n} \max_{\tau \in \mathcal{T}(Q_{1:n})} \mathbb{E}\left[ X_{\tau} \right].
\end{equation}

For $\varepsilon=\infty$, we recover $\mathcal{W}_{\infty}(\mathbf{X})=\Vn(\mathbf{X})$, and for $\varepsilon=0$, we have $\tau^* \in \argmax_{i \in [n]} \mathbb{E}[X_i]$, which yields $\mathcal{W}_0(\mathbf{X})=\max_{i \in [n]} \mathbb{E}[X_i]$. Notice that in general the mechanisms and the decision rule depend on the distribution $F_{1:n}$.
%This does not lead to a breach of privacy as $(F_i)_{i=1}^n$ are public knowledge.
\begin{comment}
At first glance, one may expect that in the high privacy regime, all the users $X_i$ would be treated similarly; this is not the case at all. In fact, what differential privacy entails is that an adversary may not change its output too much based on observed data. But if the adversary, here the online decision maker, already has access to some prior data, here the distributions $F_i$, then it can freely use it, regardless of the additional information observed $Z_i$. This contrast comes from the fact that we are dealing with a Bayesian decision problem, hence, there is no reason to treat the users identically, unlike what happens when no prior information is available.
\end{comment}
Furthermore, the quantity $\Wn(\mathbf{X})$ is non-decreasing in $\varepsilon$. Indeed, as $\varepsilon$ increases, the set of LDP mechanism increases: $\varepsilon \leq \varepsilon'$ implies $\mathcal{Q}_{\varepsilon}^{(n)} \subset \mathcal{Q}_{\varepsilon'}^{(n)}$, hence $\Wn(\mathbf{X}) \leq \mathcal{W}_{\varepsilon'}(\mathbf{X})$.
In particular, $\Wn(\mathbf{X}) \leq W_{\infty}(\mathbf{X}) = \Vn(\mathbf{X})$. In the next section, we give an efficient procedure to compute the optimal \ac{ldp}-stopping rule, as well as $\Wn(\mathbf{X})$.

\section{Optimal Backward Dynamic Programming under Local DP}
\label{sec:optimalbackward}

In this section, we compute the optimal \ac{ldp}-stopping rule and its value.
\begin{comment}
From a decision-making point of view, a simple way to construct an admissible $\varepsilon$-\ac{ldp} stopping rule is the following: 
\begin{enumerate}
    %\item Choose the message space $\mathcal{Z}_i$.
    \item Select $\varepsilon$-\ac{ldp} mechanisms $Q_i \in \mathcal{Q}_{\varepsilon}$ which generate messages $Z_i$.
    \item Select a stopping rule $\tau$ adapted to the $\mathfrak{O}_i=\sigma(Z_1,\dots,Z_i)$. 
\end{enumerate}

Hence, we want to maximise $\mathbb{E}[X_{\tau}]$ over these decision steps. 
\end{comment}
The optimal value is the solution to the joint maximization problem of~\Cref{eq:Wn}. Given $Q_{1:n} \in \mathcal{Q}_\epsilon$, we derive the optimal stopping rule $\tau(Q_{1:n})$ reaching the value
\[
\mathcal{V}_{Q_{1:n}}(\mathbf{X}) = \max_{\tau \in \mathcal{T}(Q_{1:n})} \EE[X_\tau].
\]
Then, we find the mechanisms $Q_{1:n}^*$ that maximize $\mathcal{V}_{Q_{1:n}}(\mathbf{X})$. It holds that 
\[
\Wn(\mathbf{X}) = \max_{Q_{1:n} \in \mathcal{Q}_\varepsilon^n} \mathcal{V}_{Q_{1:n}}(\mathbf{X}).
\]

\begin{comment}
Hence, we want to maximise $\mathbb{E}[X_{\tau}]$ over these three decision steps. This is a joint optimisation problem over these three decisions, and given that we can always exchange the maximums, we treat it as first choosing $\mathcal{Z}_i$, then $Q_i$, and finally the stopping rule on the $Z_i$. Regarding the choice of $\mathcal{Z}_i$, given that it does not affect the values $X_i$, it is simpler to think of this space as some kind of arbitrary alphabet, where only $\vert \mathcal{Z}_i \vert$, the alphabet size of $\mathcal{Z}_i$, matters. 
\end{comment}

\paragraph{Choosing $\tau$ given $Q_{1:n}$: a sufficient statistic for online stopping under partial observations.}
Suppose that the mechanisms $Q_{1:n}$ are already fixed. They induce a joint distribution $Q_i \otimes F_i$ on $(Z_i, X_i)$ for any $i \in [n]$. \citet{Assaf1998ASV} study the prophet inequality problem with partial observations, where similarly the online agent observes $Z_i$ and aims to maximise the stopped value with respect to $X_i$, where the $X_i$ and $Z_i$ are drawn from a given known joint distribution. They show that $\mathbb{E}[X_i \mid Z_i]$, the posterior mean of $X_i$ given $Z_i$, 
is a sufficient statistic to obtain the optimal stopping value. Specifically, Proposition $2.3$ of \citet{Assaf1998ASV} shows that an optimal stopping rule runs a \ac{bdp} algorithm on the sequence $\mathbf{Y} = (Y_i = \mathbb{E}[X_i\mid Z_i])_{i = 1}^n$. %In particular, once the mechanisms $Q_{1:n}$ are fixed, we can deduce the joint law, and the partially observed stopping problem with observable states $Z_{1:n}$ reduces to a stopping problem on $\mathbf{Y}$. 
Denoting by $(V_i(\mathbf{Y}))_{i \in [n]}$ the sequence of dynamic programming values as defined in \Cref{eq:bdp}, the above proposition gives that $\mathcal{V}_{Q_{1:n}}(\mathbf{X}) = V_1(\mathbf{Y})$ and that this value is achieved by the decision rule
\begin{equation}
    \label{eq:optstop}
\tau = \min \{i \in [n]: Y_i \geq V_{i+1}(\mathbf{Y}) \}.
\end{equation}

The optimal mechanisms $Q^*_{1:n}$ are therefore the maximizers of $V_1(\mathbf{Y})$,
\begin{equation} \label{eq:W_eps}
\Wn(\mathbf{X}) =\max_{Q_1, \dots ,Q_n \in \mathcal{Q}_\varepsilon} V_1(\mathbf{Y}).
\end{equation}
One stark difference with the setting studied in ~\citep{Assaf1998ASV}, and which is at the core of the differential privacy literature \citep{kairouz2016extremal}, is the selection of the mechanisms $Q_{1:n}$ that produce observations $Z_{1:n}$.  In~\citep{Assaf1998ASV}, the process that generates observations $Z_{1:n}$ is known but cannot be optimized over. 

In \Cref{subsec:J to W,subsec:a reduction to,subsec:optimal single step}, we give a closed-form solution to the maximization problem of \Cref{eq:W_eps}, while \Cref{subsec:on interactive} shows that the solution would remain the same if the mechanisms were allowed to be interactive, i.e. the choice of $Q_i$ can depend on $Z_{1:i-1}$.

\subsection{Locally Differentially Private Backward Dynamic Programming}
\label{subsec:J to W}

\Cref{eq:W_eps} is a maximization problem over the private mechanisms $Q_{1:n} \in \mathcal{Q}_\varepsilon^n$. The next proposition shows that the optimal mechanisms can be found greedily, by backward induction.

\begin{proposition} \label{prop:private dp}
Let $\varepsilon \geq0$, and define the sequence $(\mathscr{W}_i(\mathbf{X}))_{i \in [n]}$ as $\mathscr{W}_{n+1}(\mathbf{X})=0$ and for $i \in [n]$,
\begin{equation}\label{eq:recursionwi}
 \mathscr{W}_i(\mathbf{X})=\max_{Q_i \in \mathcal{Q}_\varepsilon} \mathbb{E}\left[\max(\mathbb{E}[X_i \mid Z_i],\mathscr{W}_{i+1}(\mathbf{X}))\right].
\end{equation}
Then, 
\begin{equation*}
\Wn(\mathbf{X})=\mathscr{W}_1(\mathbf{X}).
\end{equation*}
\end{proposition}

\begin{proof} 
%Fix $(Q_i)_{i \in [n]} \in \mathcal{Q}_{\varepsilon}^n$. Because the mechanisms are non-interactive, for any $i\in [n]$, it holds that $Y_i$ depends only on $X_i$ and $Q_i$.  By induction, it follows that for any $i \in [n]$, $V_i(\mathbf{Y})$ depends only on $(X_{i:n}, Q_{i:n})$ and that $W_i(\mathbf{X})$ depends only on $X_{i:n}$.

Call $U_i(\mathbf{X}) = \max_{Q_{i:n} \in \mathcal{Q}_\varepsilon^{n-i+1}} V_i(\mathbf{Y})$ where $\mathbf{Y} = (Y_i = \mathbb{E}[X_i\mid Z_i])_{i = 1}^n$ and $(V_i(\mathbf{Y}))_{i \in [n]}$ is the sequence of dynamic programming values as defined in \Cref{eq:bdp}. We show by induction, that for any $i \in [n]$, $U_i(\mathbf{X}) = \mathscr{W}_i(\mathbf{X})$. First, we have that $U_{n+1}(\mathbf{X}) = \mathscr{W}_{n+1}(\mathbf{X}) = 0$.
Then, fix $k < n$ and assume that $U_{k+1}(\mathbf{X}) = \mathscr{W}_{k+1}(\mathbf{X})$. We have
\begin{align*}
  U_k(\mathbf{X}) &= \max_{Q_{k:n} \in \mathcal{Q}_\varepsilon^{n-k+1}} \EE[\max(\EE[X_k | Z_k], V_{k+1}(\mathbf{Y}))] \\
  &= \max_{Q_k \in \mathcal{Q}_\varepsilon} \EE[\max(\EE[X_k | Z_k], \max_{Q_{k+1:n} \in \mathcal{Q}_\varepsilon^{n-k}} V_{k+1}(\mathbf{Y}))] \\
  &= \max_{Q_k \in \mathcal{Q}_\varepsilon} \EE[\max(\EE[X_k | Z_k], U_{k+1}(\mathbf{X}))] \\
  &= \mathscr{W}_k(\mathbf{X}).
\end{align*}
where the second equality holds because $y \mapsto \EE[\max(\EE[X_k | Z_k], y)]$ is non-decreasing. In particular, we get $\mathscr{W}_1(\mathbf{X}) = U_1(\mathbf{X}) = \Wn(\mathbf{X})$.
\end{proof}

\Cref{prop:private dp} shows that to solve \Cref{eq:W_eps}, it is enough to solve for $w \geq 0, X \sim F$ and $(Z \mid X=x) \sim Q(\cdot \mid x)$ the optimization problem
\begin{equation}
\label{eq:jw-opt}
\max_{Q \in \mathcal{Q}_\varepsilon} J_w(Q),
\end{equation}
where
\begin{equation}
\label{eq:Jw(Q)}
J_w(Q) \coloneqq \mathbb{E}\left[ \max \left(\mathbb{E}[X\mid Z],w \right) \right],
\end{equation}
which is the focus of the next two sections.

\subsection{A reduction to simple privacy preserving mechanisms}
\label{subsec:a reduction to}

\Cref{eq:jw-opt} is an optimization problem over the set of $\varepsilon$-LDP mechanisms. In this section, we show that it suffices to consider the set of binary $\varepsilon$-LDP mechanisms $\mathcal{Q}_{\varepsilon, B}$ where
\[
\mathcal{Q}_{\varepsilon, B} = \{Q \in \mathcal{Q}_\varepsilon:\mathrm{Range}(Q) = \{0, 1\} \}.
\]

\begin{proposition}[Two-point mechanisms suffice]\label{prop:two-points}
For any $\varepsilon$-\ac{ldp} mechanism $Q \in \mathcal{Q}_\varepsilon$, there exists a binary $\varepsilon$-LDP mechanism $\tilde{Q} \in \mathcal{Q}_{\varepsilon, B}$ such that $J_w(Q)=J_w(\tilde{Q})$.
\end{proposition}

\begin{proof}[Proof sketch (full proof in \Cref{app:two-point})]
    Call $\mathcal{Z}$ the range of $Q$ and consider 
    \[
     \mathcal{Z}^+=\{ z \in \mathcal{Z} : \mathbb{E}[X \mid Z=z] \geq w\}
    \]
    and 
    \[
    \mathcal{Z}^-=\{ z \in  \mathcal{Z} : \mathbb{E}[X \mid Z=z] < w\},
    \]
    where $Z, X \sim Q \otimes F$. Define $\tilde{Q}$ by $\forall x \in \mathcal{X}$,
    \begin{align*}
    &\tilde{Q}(\{1\} \mid x) = Q(\mathcal{Z}^+ \mid x) \\
    &\tilde{Q}(\{0\} \mid x) = Q(\mathcal{Z}^- \mid x).
    \end{align*}
    By direct calculation, it is then shown that $J_w(Q)$ and $J_w(\tilde{Q})$ are equal.
\end{proof}

From~\Cref{prop:two-points}, the maximum in \Cref{eq:jw-opt} can be taken over the set of binary $\varepsilon$-LDP mechanisms. In the next section, we give the maximizer in closed form.

\subsection{Solving the One-Step Problem and Completing the Optimal Policy}
\label{subsec:optimal single step}

The main result of this section is \Cref{th:optimal LDP stopping rule} that gives the optimal $\varepsilon$-\ac{ldp} stopping rule. This theorem builds on the following proposition that solves \Cref{eq:Jw(Q)} in closed form.

\begin{proposition}\label{prop:optimalchannel}
Define 
\begin{equation}
\label{eq:phi}
\phi(x,w)\coloneqq \frac{e^{\varepsilon}}{1+e^{\varepsilon}} \max(x,w)+\frac{1}{1+e^{\varepsilon}}\min(x,w), \text{ for $x,w \geq0$}.
\end{equation}

Let $X \sim F$, and $Q \mapsto J_w(Q)$ the objective in \Cref{eq:Jw(Q)}. 

The mechanism
\begin{equation}
\label{eq:Rwe}
R_{w,\varepsilon}(\, \cdot  \mid x) \coloneqq \begin{cases}
\mathrm{Ber}\left(\frac{e^{\varepsilon}}{1+e^{\varepsilon}}\mathds{1}[x \geq w]+ \frac{1}{1+e^{\varepsilon}} \mathds{1}[x<w]\right),& \text{ if } \mathbb{E}[\phi(X,w)]\geq \max(\mathbb{E}[X],w),\\
1, & \text{ if } \mathbb{E}[X] > \max(w,\mathbb{E}[\phi(X,w)]),\\
0, & \text{ if } w \geq \max(\EE[X],\mathbb{E}[\phi(X,w)]),
\end{cases}
\end{equation}
achieves the maximum of $J_w$ and
\begin{equation}
\label{eq:Jw}
J_w(R_{w, \varepsilon}) = \max(\EE[X], w, \EE[\phi(X, w)]).
\end{equation}
\end{proposition}

\begin{proof}[Proof sketch (full proof in \Cref{app:optimalchannel})] By \Cref{prop:two-points}, it is sufficient to consider two-point mechanisms, which are entirely described by the conditional law $\varphi(x)=\Pr(Z=1 \mid X=x)$ since $\Pr(Z=0 \mid X=x)=1-\varphi(x) $.
The privacy constraints are equivalent to the constraints $(\inf \phi, \sup \phi) \in \mathfrak{C}$ where
\[
\mathfrak{C} = \{a, b \in [0, 1]: (1-a) \leq e^{\varepsilon}(1-b), b \leq e^{\varepsilon} a, a \leq b \}
\]
Then, we show that $\varphi^*$ reaching $\max_{Q \in \mathcal{Q}_\varepsilon} J_w(Q)$ is solution to
\[
\max_{\varphi: (\inf \varphi, \sup \varphi) \in \mathfrak{C}} w + \int_{\mathbb{R}_+} (x-w) \varphi(x) dF(x),
\]
which implies that $\varphi^*$ is of the form
\[
\varphi^*(x) =
\begin{cases}
    a & \text{if } x < w \\
    b & \text{otherwise}.
\end{cases}
\]
where $(a, b) \in \mathfrak{C}$. 
The problem therefore, reduces to finding the optimal values for $a$ and $b$.
The set of constraints $\mathfrak{C}$ forms a system  of linear inequalities in $a$ and $b$ that has $3$ extreme points, either $(a,b)=(0,0)$, $(a,b)=(1,1)$, or $(a, b) = (\frac{1}{1+e^{\varepsilon}}, \frac{e^{\varepsilon}}{1+e^{\varepsilon}})$.
So the maximum is reached for one of these three points. Choosing $a=b=0$ or $a=b=1$ recovers the non-informative mechanism $Q( \{ 0 \} \mid x) = 1$ or $Q( \{ 1 \} \mid x) = 1$ and gives objective value 
$
J_{w}(Q)=\max(\mathbb{E}[X],w)
$.
The other extreme point corresponds to the binary randomized response  
\[
Q( \cdot \mid x) = \mathrm{Ber}\left(\frac{e^{\varepsilon}}{1+e^{\varepsilon}}\mathds{1}[x \geq w]+ \frac{1}{1+e^{\varepsilon}} \mathds{1}[x<w]\right),
\]
and gives the value  $J_{w}(Q)=\mathbb{E}[\phi(X,w)]$.
\end{proof}

% The mechanism $R_{w, \varepsilon}$, highlighted in~\Cref{prop:optimalchannel}, is a randomized response mechanism applied to the input $\ind{x > w}$. Note that when the condition $\EE[\phi(X, w)] \geq \max(\EE[X], w)$ is not satisfied, the mechanism can be chosen arbitrarily since its output will be ignored. The convention chosen in~\Cref{eq:Rwe} is designed such that the output $Z = 1$ indicates that $\EE[X | Z=1] \geq \max(\EE[X], w)$.

% The mechanism $R_{w,\varepsilon}$ in \Cref{prop:optimalchannel} admits a simple interpretation.

% The one-step objective $J_w(Q)=\EE[\max(\EE[X\mid Z],w)]$ depends on $Q$ only through which reports lead the decision maker to ``accept'' (i.e., have posterior mean at least $w$) versus ``reject''.
% Thus, it is optimal to use the report $Z$ only to convey coarse information about whether $X$ lies above or below the continuation threshold $w$.

The mechanism $R_{w,\varepsilon}$ is exactly randomized response applied to the binary predicate $\ind{x\ge w}$:
values above $w$ are more likely to trigger the ``accept'' message, and values below $w$ are more likely to trigger ``reject'',
with likelihood ratio bounded by $e^{\varepsilon}$ as required by $\varepsilon$-LDP.
The third quantity in \Cref{eq:Jw}, namely $\EE[\phi(X,w)]$, is precisely the one-step value obtained when the decision maker uses this randomised response signal.
When $\EE[\phi(X,w)]$ fails to beat the trivial baseline $\max(\EE[X],w)$, privacy makes the threshold signal too weak to be useful,
and an optimal mechanism can be chosen to be non-informative (a constant report), in which case the message is ignored and the decision reduces to comparing $\EE[X]$ to $w$.
The convention in \Cref{eq:Rwe} is chosen so that $Z=1$ always corresponds to the ``accept'' outcome in the subsequent stopping rule, i.e. $\EE[X | Z=1] \geq \max(\EE[X], w)$.

We conclude this section by giving, in the next theorem, the optimal $\varepsilon$-LDP stopping rule.

\begin{theorem}
    \label{th:optimal LDP stopping rule}
    Let $\mathbf{X} \in \mathcal{C}_n$ with joint distribution $F_1\otimes \dots \otimes F_n$ and let $(W_i(\mathbf{X}))_{i \in [n+1]}$ be a sequence such that $W_{n+1}(\mathbf{X}) = 0$ and for $i \in [n]$,
    \begin{equation}
    \label{eq:wphi}
    W_i(\mathbf{X}) = \max(\EE[X_i], W_{i+1}(\mathbf{X}), \EE[\phi(X_i, W_{i+1}(\mathbf{X}))])
    \end{equation}
    where $\phi$ is defined in \Cref{eq:phi}.
    
    The mechanisms $(Q_i^* = R_{W_{i+1}(\mathbf{X}), \varepsilon})_{i \in[n]}$ where for any $w \geq 0$, $R_{w, \varepsilon}$ is defined in \Cref{eq:Rwe} together with the stopping rule 
    \[
    \tau^* = \min \{i \in [n]: Z_i = 1 \}
    \]
    where $(Z_i \mid X_i = x) \sim Q_i^*( \cdot \mid x)$  are solution to the maximization problem \Cref{eq:Wn}. Furthermore $\Wn(\mathbf{X}) = W_1(\mathbf{X})$.
\end{theorem}
\begin{proof}
    By induction, the sequence $(W_i)_{i \in [n]}$ of this theorem coincides with the sequence of dynamic programming values  $(\mathscr{W}_i)_{i \in [n]}$ of \Cref{prop:private dp}. First, $W_{n+1} = \mathscr{W}_{n+1} = 0$. Then, for $k \in [n]$, assuming $W_{k+1} = \mathscr{W}_{k+1}$ we have by \Cref{prop:optimalchannel} that the max over $\varepsilon$-LDP mechanisms in \Cref{eq:recursionwi} is reached by $R_{\mathscr{W}_{i+1}, \varepsilon} = R_{W_{i+1}, \varepsilon}$ showing that $W_k = \mathscr{W}_k$.
    Lastly, by \Cref{eq:optstop}, and since $Q_i = R_{W_{i+1}, \varepsilon}$, the optimal stopping rule is then
    $
    \tau^* = \min \{i \in [n], \EE[X_i | Z_i] \geq W_{i+1}, \}$
    which by definition of $R_{W_{i+1}, \varepsilon}$ is the same as $
    \tau^* = \min \{i \in [n], Z_i = 1 \}$.
\end{proof}

Theorem~\ref{th:optimal LDP stopping rule} yields a transparent ``privatize-then-decide'' description of optimal behavior under LDP.
The backward recursion \Cref{eq:wphi} plays the role of a continuation-value dynamic program:
$W_{i+1}(\mathbf{X})$ is the optimal value of continuing past time $i$ under LDP, and therefore acts as the acceptance threshold at stage $i$.
Given this threshold, the optimal reporting mechanism is \emph{binary}: agent $i$ locally runs randomized response on whether $X_i$ exceeds $W_{i+1}(\mathbf{X})$,
producing a bit $Z_i$ that is informative about the event $\{X_i \ge W_{i+1}(\mathbf{X})\}$ but bounded by the privacy budget.
The decision maker then follows the simplest possible stopping rule: stop at the first time $i$ such that $Z_i=1$.
In other words, local privacy collapses each arrival to a one-bit ``recommendation'' about whether accepting is better than continuing,
and the optimal policy is obtained by combining (a) the dynamic computation of thresholds $(W_i)_{i\in[n]}$ with (b) optimal one-step privatization at each threshold.

% The optimal rule in \Cref{th:optimal LDP stopping rule} prescribes the decision maker to stop at step $i$ if the output of the privacy mechanism indicates that $\{X_i \geq W_{i+1}\}$ is more likely than $\{ X_i < W_{i+1} \}$ where the sequence of thresholds $(W_i)_{i \in [n]}$ are dynamically computed.

% We now want to exhibit the optimal binary mechanisms and a simple recurrence relation which yields the optimal $\varepsilon$-\ac{ldp} stopping value $W^{(n)}_\varepsilon$. It is likely that a backward dynamic approach will yield the optimal value, and we show that once the optimal mechanisms are fixed, this is indeed the case. 

% First, for comparison, let us look at the standard optimal stopping value 
% \begin{equation}
% \mathrm{V}^{(n)}(\mathbf{X}) \coloneqq \sup_{\tau \in \mathcal{T}_n} \mathbb{E}[X_{\tau}],
% \end{equation}
% where $\mathcal{T}_n=\mathcal{T}((\sigma(X_1,\dots,X_i))_{i \in [n]})$ is the set of stopping rules with respect to the $X_i$ observations filtration. In \cite{Chow1971GreatET}, it is shown for independent positive $X_i$ that the optimal stopping time $\tau^*$ exists, and is obtained through backward dynamic programming. More specifically, let us denote by $V_i$ the optimal expected value of being at state $i$, we have the following recurrence:
% \begin{equation}
%     V_i=\mathbb{E}[\max(X_i,V_{i+1})], 
% \end{equation}
% with $V_{n+1}=0$ (or equivalently $V_n=\mathbb{E}[X_n]$), and $\mathrm{V}=V_1$. Remark that allowing for randomization in the stopping rule doesn't change the value in this setting. 

\subsection{On Interactive Locally Differentially Private Mechanisms}
\label{subsec:on interactive}

%As \Cref{prop:optimalchannel} we can therefore simply optimize over those binary mechanisms to obtain $W^{(n)}_{\varepsilon}$. As we will see later in \hr{ref missing} [REF HERE], the fact that binary mechanisms are optimal only holds for the online decision maker: if all the data $Z_i$ could be observed at once, then binary mechanisms are not necessarily optimal. We next compute the optimal \ac{ldp}-channel for this functional $J$.

Up until now, we have assumed that the mechanisms $Q_{1:n}$ are chosen a priori, before observing the messages $Z_i$. A closely related notion in differential privacy, when a stream of data is observed, is to consider interactive local differential privacy \citep{joseph2019role}. This allows the decision maker to choose a mechanism for time $i$ by leveraging the messages observed at times $j<i$. Concretely, the kernel $Q_i$ can generate a message $Z_i$ conditionally on $(Z_1,\dots,Z_{i-1},X_i)$.  
The next proposition proves that allowing for interactive \ac{ldp} mechanisms does not improve the optimal stopping value.

\begin{proposition}
\label{prop:interactive}
Suppose that we allow for each mechanism $Q_i$ to depend on the history $H_i=(Z_1,\dots,Z_{i-1})$, and denote by $\mathcal{W}_{\varepsilon}^{(n),I}(\mathbf{X})$ the (weakly larger) optimal stopping value for this enlarged space of mechanisms. This provides no advantage compared to non-interactive mechanisms: 
\begin{equation*}
\mathcal{W}_{\varepsilon}^{(n),I}(\mathbf{X})=\Wn(\mathbf{X}).
\end{equation*}
\end{proposition}

\begin{proof}[Proof sketch (Full proof in \Cref{app:interactive})]
Fix $(Q_i)_{i \in [n]}$ a set of interactive privacy mechanisms and $Z_i$ the associated random variables. The value attained by the optimal stopping rule can be computed by dynamic programming. It is given by $N_1$ where $N_{n+1} = 0$ and for $1 \leq k \leq n$, 
\[
N_{k} =\EE[\max(\EE[X_k | Z_{1:k}], N_{k+1}) \mid Z_{1:k-1}].
\]
Then defining $(W_k)_{k \in [n]}$ as in \Cref{th:optimal LDP stopping rule},
we show by induction that $N_k \leq W_k$ for any $1 \leq k \leq n$.

\end{proof}

\section{Competitive Ratio Against Non-private Benchmarks}
\label{sec:competitive}

\Cref{th:optimal LDP stopping rule} gives a useful characterization of $\Wn(\mathbf{X})$. In this section, this result will be leveraged to obtain worst-case competitive ratio guarantees, when comparing against the optimal value of the non-private online agent $\Vn(\mathbf{X})$ and the optimal value of the prophet $\Mn(\mathbf{X}) \coloneqq \mathbb{E}[\max_{ i \in [n]} X_i]$.

\subsection{Competitive Ratio Against Non-Private Optimal Stopping}

First, we study the relative value loss of an online agent that must respect \ac{ldp} constraints compared to one without such constraints. We derive the exact tight competitive ratio in the worst-case as a function of $n$ and $\varepsilon$.

\begin{theorem}\label{thm:cr_stopping_ldp}
For all $n \in \mathbb{N}^*$ and $\varepsilon\geq0$,
\begin{equation}
\inf_{\mathbf{X} \in \mathcal{C}_n} \frac{W_{\varepsilon}^{(n)}(\mathbf{X})}{\Vn(\mathbf{X})} = \frac{e^{\varepsilon}}{n-1+e^{\varepsilon}}.
\end{equation}
\end{theorem}

We will prove this theorem by first exhibiting a lower bound that does not depend on the $X_i$, and then showing that the ratio attains this bound for a specific sequence of $X_i$. The following proposition gives a uniform lower bound.
\begin{proposition}
\label{prop:lb-w/v}
For all $n \in \mathbb{N}^*$, $\varepsilon\geq0$, and $\mathbf{X} \in \mathcal{C}_n$, we have
\begin{equation*}
W_{\varepsilon}^{(n)}(\mathbf{X}) \geq \frac{e^{\varepsilon}}{n-1+e^{\varepsilon}} \cdot  \Vn(\mathbf{X}).
\end{equation*}
\end{proposition}
\begin{proof}[Proof sketch (full proof in \Cref{app:lb-w/v})]
The proof considers that the sequence $(W_i(\mathbf{X}))_{i \in [n]}$ defined in \Cref{th:optimal LDP stopping rule} and the sequence $(V_i(\mathbf{X}))_{i \in [n]}$ defined in \Cref{eq:bdp}. It is shown that for any $i \in [n]$,
\[
\frac{V_i(\mathbf{X})}{W_i((\mathbf{X}))} \leq \frac{V_{i+1}(\mathbf{X})}{W_{i+1}((\mathbf{X}))} + \exp(-\varepsilon).
\]
\end{proof}

We next prove the corresponding upper bound. 

\begin{proposition}
\label{prop:ub-w/v}
Consider the following sequence: $X^q_n=1$ a.s. and for $i \in [n-1]$,
\begin{equation*}
X^q_i= \begin{cases}
    1+e^{-\varepsilon}\frac{1-q}{q}, \quad \text{with probability $q$},\\
    0, \quad \text{with probability $1-q$}.
\end{cases}
\end{equation*}
Then,
\begin{equation*}
\lim_{q \rightarrow 0} \frac{\Wn(\mathbf{X}^q)}{\Vn(\mathbf{X}^q)}=\frac{e^{\varepsilon}}{n-1+e^{\varepsilon}}.
\end{equation*}
\end{proposition}
\begin{proof}[Proof sketch (Full proof in \Cref{app:ub-w/v})]
At step $i$, the non-private decision maker stops at step $i$ according to a binary signal $\ind{X_i > V_{i+1}}$. The private decision maker also observes a binary signal, but corrupted by the randomized response mechanism. A positive signal can be observed either because of a high value or because of the privacy noise (false positive). The first $n-1$ variables $(X_i^q)_{i \in [n-1]}$ each have high value with low probability and zero value otherwise. 
% However, they are taken in such a way that the private decision maker never chooses them. 
The risk of having zero value because of a false positive signal slightly outweighs the potential gains, and we get \[
\Wn(\mathbf{X}^q) = 1.\]
However, for the non-private decision maker, there are no false positives, which allows her to take a large value if it appears. Computations detailed in appendix show that \[
\lim_{q \rightarrow 0} \Vn(\mathbf{X}^q) = (n-1) \exp(-\varepsilon) + 1.\]
\end{proof}

Together \Cref{prop:ub-w/v} and \Cref{prop:lb-w/v} imply \Cref{thm:cr_stopping_ldp}. A direct corollary of~\Cref{thm:cr_stopping_ldp} is the following characterization of the phase transition of $\varepsilon$ as a function of $n$, that yields non-trivial (not $1$ nor $0$) asymptotic tight bounds.

\begin{corollary}
    Let $\gamma >0$ and $\varepsilon=\log(\gamma n)$, with $n \geq 1/\gamma$. Then
    \begin{equation*}
    \lim_{n \rightarrow \infty} \inf_{\mathbf{X} \in \mathcal{C}_n}  \frac{\Wn(\mathbf{X})}{\Vn(\mathbf{X})} = \frac{\gamma}{\gamma+1}.
    \end{equation*}
\end{corollary}

For specific applications, it might be useful for the decision maker to specify a per-user privacy budget $\varepsilon_i$. We can obtain the same result for this more refined privacy analysis. The proof is presented in~\Cref{app:distinct privacy} and is essentially identical to that of~\Cref{thm:cr_stopping_ldp}. 

\begin{corollary}
\label{cor:distinct privacy}
For all $n \in \mathbb{N}$ and possibly distinct privacy levels $\boldsymbol{\varepsilon}=(\varepsilon_1,\dots,\varepsilon_n)\geq0$,
\begin{equation}
\inf_{\mathbf{X} \in \mathcal{C}_n} \frac{\Wn(\mathbf{X})}{\Vn(\mathbf{X})} = \frac{1}{1+\sum_{i \in [n-1]} e^{-\varepsilon_i}}.
\end{equation}
\end{corollary}

Remark that this competitive ratio does not depend on the last privacy budget $\varepsilon_n$: it is always optimal to stop at the last time step if reached, regardless of its value.

% \paragraph{Simple Stopping Rules} \textcolor{red}{Add section on simple stopping rules} One classical result of prophet inequalities, is that the competitive ratio remains $1/2$ even when restricting to single-threshold policy, that is to say stopping rules which correspond to $\tau=\min\{ i \in [n] : X_i \geq T\}$ for some threshold $T$. The next proposition shows that the same lower bound can be achieved for the private online agent versus the prophet using a single threshold. However, \ac{ldp}-channels are still distinct. (Could we have a negative result if we enforce same mechanism and single threshold?). \mathieu{on mettra la preuve en appendix je pense} \\

\subsection{Competitive Ratio Against Non-Private Prophet}

Now, we want to compare the value of the optimal \ac{ldp} stopping rule with that of the standard non-private prophet. 

\begin{proposition} \label{prop:wagainstm}
    For $n \geq 2$ and $\varepsilon \geq0$, we have 
    \begin{equation*}
       \frac1{n} \vee \frac{e^{\varepsilon}}{2e^{\varepsilon} - 2 + 2n} \leq \inf_{\mathbf{X} \in \mathcal{C}_n} \frac{\Wn(\mathbf{X})}{\Mn(\mathbf{X})}  \leq \frac{e^{\varepsilon}}{2e^{\varepsilon} - 2 + n}.
    \end{equation*}
\end{proposition}

\begin{proof}[Proof sketch (full proof in \Cref{app:wagainstm})]
First note that by using the $0$-\ac{ldp} deterministic stopping rule $\tau=\argmax_{i \in [n]} \mathbb{E}[X_i]$ achieves at least a $1/n$ competitive ratio against the non-private prophet.  Indeed, 
\begin{align*}
    \mathbb{E}[\max_{i \in [n]} X_i] &\leq \sum_{i \in [n]} \mathbb{E}[X_i] \\
    &\leq n\max_{i \in [n]}\mathbb{E}[X_i].
\end{align*}
The lower bound $\frac{e^{\varepsilon}}{2e^{\varepsilon} - 2 + 2n}$ is immediate from \Cref{thm:cr_stopping_ldp} and the classical prophet lower bound $\Vn(\mathbf{X})/\Mn(\mathbf{X}) \geq \frac12$, by writing out $\Wn(\mathbf{X})/\Mn(\mathbf{X})=(\Wn(\mathbf{X})/\Vn(\mathbf{X}))\cdot (\Vn(\mathbf{X})/\Mn(\mathbf{X}))$.

For the upper bound, we will use a similar counter-example as in \Cref{thm:cr_stopping_ldp}. Use $X_i \sim X^q$ for $i=2,\dots,n-1$, $X_n=1/c$ with probability $c \in [0,1]$ and $X_n=0$ with probability $1-c$, and $X_1=1$ a.s. Computations displayed in the appendix show that \[
\Wn(\mathbf{X}) = 1, \]
and that if $c = q^2$, it holds that 
\[\lim_{q \rightarrow 0} \Mn(\mathbf{X}) = 2 + \exp(-\varepsilon) (n-2).
\] \end{proof}

% While for arbitrary $\varepsilon$ the tight constant $\inf_{\mathbf{X} \in \mathcal{C}_n} \Wn(\mathbf{X})/\Mn(\mathbf{X})$ is unknown, for $\varepsilon=0$ the upper-bound of \Cref{prop:wagainstm} is tight as it matches the $1/n$ lower bound. This may suggest that the tight competitive ratio is in fact $e^{\varepsilon}/(2e^{\varepsilon} - 2 + n)$. For $\varepsilon = +\infty$, the upper bound and lower bound become $\frac12$, matching the classical result in non-private prophets.

The upper and lower bounds of \Cref{prop:wagainstm} match exactly in the asymptotic regimes in $\varepsilon$: (a) for $\varepsilon = 0$, they both reduce to $1/n$, while (b) for $\varepsilon = \infty$, they both reduce to $1/2$. For finite $\varepsilon > 0$, characterizing the tight $\inf_{\mathbf{X} \in \mathcal{C}_n} \Wn(\mathbf{X})/\Mn(\mathbf{X})$ is an interesting open problem. 

\subsection{Simpler Private Stopping Rules}

In the prophet inequality literature, special attention has been given to simple policies when comparing to the prophet benchmark, in particular, single-threshold policies which for some threshold $T \geq 0$ correspond to the stopping rule $\tau(T) \coloneqq \min\{ i \in [n] : X_i \geq T\}$ which accepts the first item with a value $X_i$ above the threshold $T$.  \citet{SamuelCahn1984} famously shows that the classical competitive ratio of $\Vn(\mathbf{X})/\Mn(\mathbf{X}) \geq 1/2$ can be achieved through the single threshold policy $T=\mathrm{median}(\max_{i \in [n]} X_i)$. Hence, this raises the question of whether similar simple policies can achieve a good competitive ratio in the private setting. The stopping rule of \Cref{th:optimal LDP stopping rule} used on reported messages can already be considered a single threshold, as we stop at the first $Z_i \geq 1$. However, to compute the $(Z_i)$, it is required to explicitly compute the backward induction thresholds. The next proposition proves that simpler mechanisms achieve the lower bound of \Cref{prop:wagainstm} by leveraging single threshold policies in the non-private regime.

\begin{proposition} 
\label{prop: private T / M}
Let $T=\mathrm{median}(\max_{i \in [n]} X_i)$ (or any single threshold rule that achieves a $1/2$ competitive ratio) and  consider the sequence of mechanisms $Q_i$ defined as, $Q_n = 1$, and for $i \leq n - 1$ and $x\geq 0$, \[
Q_i(\,\cdot \mid x) =\mathrm{Ber}\left( \frac{e^{\varepsilon}}{n-i+e^{\varepsilon}}\mathds{1}[x\geq T]+\frac{1}{n-i+e^{\varepsilon}}\mathds{1}[x< T] \right).
\]
Let $\tau_{\varepsilon}=\min \{ i\in[n] : Z_i=1\}$ where $Z_i \sim Q_i(\cdot | x)$. Then, $Q_i \in \mech_\varepsilon$ and
\begin{equation*}
\mathbb{E}[X_{\tau_{\varepsilon}}] \geq \frac{e^{\varepsilon}}{n-1+e^{\varepsilon}} \cdot \mathbb{E}[X_{\tau(T)}] \geq \frac{e^{\varepsilon}}{n-1+e^{\varepsilon}} \cdot \frac{1}{2} \cdot \mathbb{E}[\max_{i \in [n]} X_i].
\end{equation*}
\end{proposition}
\begin{proof}[Proof sketch (full proof in \Cref{app:private T / M})]
Let $\alpha_i=\frac{1}{n-i+e^{\varepsilon}}$ and $\beta_i=\frac{e^{\varepsilon}}{n-i+e^{\varepsilon}}$.
We have
\[
\EE[X_{\tau_\varepsilon}] = \sum_{i \in [n]} \EE[X_i \ind{\tau_\varepsilon = i}].
\]
Remark that $\ind{\tau_\varepsilon =i}$ is implied by the conjunction of $Z_1 =0 \wedge \dots \wedge Z_{i-1}=0 \wedge Z_i=1$ and $X_1 < T \wedge \dots \wedge X_{i-1} < T \wedge X_i \geq T$. It follows that
\[
\EE[\ind{\tau_\varepsilon = i} \mid X_1, \dots ,X_n] \geq \beta_i \ind{X_i \geq T} \prod_{j=1}^{i-1} (1 - \alpha_j) \ind{X_j < T}.
\]
The sequence $(\alpha_i)_{i \in [n]}$ and $(\beta_i)_{i \in [n]}$ are chosen such that 
$\beta_i \prod_{j=1}^{i-1} (1 - \alpha_j) = \frac{e^\varepsilon}{n-1 + e^\varepsilon}$,
so we get
$
\EE[X_{\tau_\varepsilon}] \geq \sum_{i=1}^n \frac{e^\varepsilon}{e^\varepsilon + n - 1} \EE[X_i \ind{\tau(T) = i}] = \frac{e^\varepsilon}{e^\varepsilon + n - 1} \EE[X_{\tau(T)}]$,
which gives the first inequality. The second inequality comes from the classical prophet inequality.
\end{proof}

The mechanism in \Cref{prop: private T / M} is computationally efficient, but time-varying. Whether good relative performance can be achieved with a single static mechanism across all time steps is left as an open problem.  Comparing the online \ac{ldp} decision maker with a non-private benchmark is informative in terms of utility privacy trade-offs.  However, \Cref{prop:ub-w/v} implies that 
\[
\inf_{n \in \NN, \mathbf{X} \in \mathcal{C}_n} W^{(n)}_\varepsilon(\mathbf{X}) / \Vn(\mathbf{X}) = 0,
\]
which shows that this ratio cannot be used to compare the performance of private online algorithms in the worst-case.  In the next section, we compare the private online decision maker with a private prophet, therefore characterizing the difficulty of online learning under privacy constraints.

\section{Competitive Ratio Against Private Prophet}
\label{sec:competitive ratio}
In this section, we now suppose that both the online decision maker and the prophet can only observe $\varepsilon$-\ac{ldp} messages from some channels. The main question we aim to answer is how the relative performance evolves with the privacy parameter $\varepsilon$.
%Usually, for almost all decision problems, utility tends to decrease as privacy increases. Here, because we are evaluating relative performance between two decision processes, we actually end up with the problem ``becoming easier'' as privacy increases, going against the naive intuition.
%\hr{The fact that the problem becomes trivial in the limit $\varepsilon \rightarrow 0$ is intuitive (both algorithms cannot do anything). What is not trivial is how the competitive ratio evolves with the privacy parameter.}

\subsection{Optimal Mechanisms for the Private Prophet}
For $\mathbf{X} \in \mathcal{C}_n$, we study the value $\Ln(\mathbf{X})$ reached by the optimal private prophet. The decision-making of the prophet consists of choosing $Q_{1:n} \in \mathcal{Q}_\epsilon^n$ that produce the signals $Z_{1:n}$ from $X_{1:n}$ and then choosing an index $s \in [n]$ based on $Z_{1:n}$.
\begin{comment}
Given $X_1 \sim F_1, \dots, X_n \sim F_n$, the decision-making process for the prophet now consists of: 
\begin{enumerate}
 \item Selecting $n$ \ac{ldp} mechanisms $Q_1,\dots ,Q_n \in \mathcal{Q}_{\varepsilon}$, generating $Z_i \mid X_i = x \sim Q_i( \cdot \mid x)$,
 \item Selecting a random index $s \in [n]$ that is $\sigma(Z_1,\dots,Z_n)$ measurable.
 \item The prophet then receives expected utility $\mathbb{E}[X_s]$.
\end{enumerate} 
\end{comment}
By Proposition $2.2$ of \citep{Assaf1998ASV}, the optimal index selection is $s \in  \argmax_{i \in [n]} \mathbb{E}[X_i \mid Z_i]$. Hence, the goal of the prophet is to find the best private channels $Q_1, \dots, Q_n$ that maximize
 \begin{equation}
 \label{eq:Lq}
L(Q_1,\dots,Q_n)\coloneqq \mathbb{E}\left[\max_{i \in [n]} \mathbb{E}[X_i \mid Z_i]\right].
\end{equation}
We define the optimal value of the private prophet as
 \begin{equation}
 \label{eq:Ln}
\Ln(\mathbf{X})\coloneqq\max_{(Q_1,\dots, Q_n) \in \mathcal{Q}_{\varepsilon}^{n}} L(Q_1,\dots,Q_n). 
\end{equation}

We show that under some assumptions, the maximum is reached for staircase mechanisms as defined in~\cite{kairouz2014extremal}. The following assumption is only used in this subsection.

\begin{assumption}[Finite support]
For any $i \in [n]$, $F_i$ is of finite support.
\end{assumption}
This assumption implies that there exists $m \in \NN$ such that $\mathcal{X} = \{v_1, \dots, v_m\}$ where $0 \leq v_1 \leq \dots \leq v_m$. Secondly, we only consider $\epsilon$-LDP mechanisms with a discrete output, i.e. $Q \in \mathcal{Q}_\varepsilon^\textrm{disc}$ where
\[
\mathcal{Q}_\varepsilon^\textrm{disc} \coloneqq \{Q \in \mathcal{Q}_\varepsilon : \textrm{Range}(Q) \subset \NN \}.
\]
These assumptions are made for simplicity and to fit the setting of \cite{kairouz2014extremal}. 

%For the prophet, the decision-making problem boils down to finding the jointly optimal vector of \ac{ldp} mechanisms that maximize a utility function, here the expectation of the maximum of the posterior means. \citet{kairouz2016extremal} provides a characterization of such mechanisms.

%Note that if the supports are distinct, this only multiplies the total support size by $n$, so for simplicity we keep the same support. \hr{It is unclear why you need this assumption. To me, it seems that you can approximate the distribution of all variables by a histogram and obtain an arbitrarily close approximation.}

First, we show in \Cref{lemma:sublinear} that each coordinate function of $L$ defined in \Cref{eq:Lq} can be decomposed as a sum of sub-linear functions. The proof is given in \Cref{app:sublinear}.
\begin{lemma}\label{lemma:sublinear}
For $i \in [n]$, and $Q_1,\dots,Q_{i-1},Q_{i+1},\dots,Q_n \in \mathcal{Q}_\varepsilon^\textrm{disc}$, the functional
\[
L_i: Q_i \in \mathcal{Q}_\varepsilon^\textrm{disc} \mapsto L(Q_1,\dots,Q_n)
\]
satisfies 
\[
L_i(Q_i) = \sum_{z \in \mathrm{Range}(Q_i)} \mu(Q_{i, z})
\]
where $Q_{i, z} = (Q_i(z | v_1), \dots, Q_i(z | v_m))$ and $\mu$ is sub-linear that is $\forall q, q' \in \RR^m$ and $\lambda \in [0,1]$,
\begin{align*}
\mu(\lambda q)&=\lambda \mu(q) && \text{(homogeneity)}\\
\mu(q + q')& \leq \mu(q)+\mu(q') && \text{(sub-additivity)}. 
\end{align*}
\end{lemma}
Combining \Cref{lemma:sublinear} with Theorem $2$ of~\cite{kairouz2014extremal} implies that there exists optimal staircase mechanisms.

\begin{proposition}\label{prop:staircase}
For $\varepsilon \geq 0$, there exists \ac{ldp} mechanisms $Q_i^*$ which maximize $L(Q_1,\dots,Q_n)$ where 
\begin{enumerate}
    \item For all $i \in [n]$, the alphabet size $\vert \textrm{Range}(Q_i) \vert $ can always be taken smaller than $\vert \mathcal{X} \vert $ the original data alphabet size, with no loss of utility.
    \item The prophet can always only use staircase mechanisms with no loss of utility, that is to say, for $i \in [n]$ mechanisms $Q_i \in \mathcal{Q}_{\varepsilon}^\textrm{disc}$ which satisfy the following constraint:
    \begin{equation*}
        \left \vert \log\left( \frac{Q_i(z\mid x)}{Q_i(z \mid x')}  \right) \right \vert \in \{0,\varepsilon\}, \quad \text{for all $x,x' \in \mathcal{X}$ and $z \in \textrm{Range}(Q_i)$}.
    \end{equation*}
\end{enumerate}
\end{proposition}
\begin{proof}
Consider the sequence of optimal mechanisms $Q_1^*,\dots,Q^*_n$. Let $ i \in [n]$. At optimality, since the joint feasible set of the $(Q_i)_{i \in [n]}$ is the simple product of $\mathcal{Q}_{\varepsilon}^\textrm{disc}$, $Q_i^*$ must maximize the coordinate function $L_i: Q_i \mapsto L(Q_1,\dots,Q_n)$. This function is sub-linear by \Cref{lemma:sublinear}, and applying Theorem $2$ of \citet{kairouz2016extremal}, we can conclude that there exists at least another $\tilde{Q}_i$ which satisfies statement $1)$ and $2)$ of the proposition for this fixed $i$ and which achieves the same value. Iterating over all $i \in [n]$ concludes the proof.
\end{proof}

\begin{comment}
\hr{We should discuss about this next paragraph, if we lack space, let's remove it otherwise, we need to carry experiments.}
While the optimal mechanisms for the prophet are not completely unstructured, they nonetheless can be more complex than that of the online decision maker \Cref{prop:optimalchannel}: some simple numerical experiments \mathieu{[REF HERE]} clearly show that binary mechanisms are not optimal in general. 
% \end{proposition}
\mathieu{I think it might be good to add one or two examples no?}
% We define the functional 
\end{comment}

While \Cref{prop:staircase} gives necessary conditions satisfied by the optimal mechanism, it does not give an efficient algorithm to compute the value of the private prophet. We leave this question for future work. In the next sub-section, we give the ratio of the performance reached by optimal private online decision maker and that of the private prophet in the worst-case. One minor consequence of this result is that the value of the prophet can be approximated efficiently by the optimal private online decision-maker.

\subsection{Tight Competitive Ratio}

In this section, we compute the value of the following worst-case competitive ratio, 
\begin{equation*}
C_{\varepsilon, n} \coloneqq  \inf_{\mathbf{X} \in \mathcal{C}_n} \frac{\Wn(\mathbf{X})}{\Ln(\mathbf{X})}.
\end{equation*}
In \cite{Assaf1998ASV},  they show in Theorem $2.1$ that the optimal competitive ratio of the online agent relative to the prophet for partial observations $(X_i,Z_i)$ remains $1/2$. This immediately yields a lower bound of $C_{\varepsilon, n} \geq 1/2$ in the setting we study: it suffices for the online decision maker to select the same optimal mechanisms as the prophet, and then run \ac{bdp} on the partial observations. In that case, the prophet and the online decision maker have the same information structure, and Theorem $2.1$ of \cite{Assaf1998ASV} applies.

In our case, comparing the performance of the optimal private online decision maker and that of the optimal prophet is a different problem for two reasons. First, the online agent is allowed to use distinct privacy mechanisms, which will, at optimality, provide something at least weakly better than using the same mechanisms as the prophet. Second, not all information structures are admissible. If we suppose for a moment that both the online decision-maker and the prophet use the same mechanisms, then, as highlighted in \citet{Assaf1998ASV} this is equivalent to a standard prophet problem on $\mathbb{E}[X_i \mid Z_i]$. But not all distributions over $\mathbb{R}_+$ can be realized through the posterior mean obtained via an $\varepsilon$-\ac{ldp} mechanism. For instance for $\varepsilon <\infty$, there does not exist a mechanism $Q \in \mathcal{Q}_{\varepsilon}$ and a positive random variable $X$ such that $\mathbb{E}[X \mid Z]=0$ with probability $1-\delta$, and $\mathbb{E}[X \mid Z]=1/\delta$ with probability $\delta$, for $\delta>0$ sufficiently small which is the typical counter-example used to prove the $1/2$ upper bound. In other words, the space of admissible distributions is restricted to those that can be realized via some $\mathbb{E}[X_i \mid Z_i]$. As $\varepsilon$ decreases, the space of admissible distributions becomes more constrained, which suggests that $C_{\varepsilon, n}$ decreases with $\varepsilon$, an intuition that is shown to be true in the following \Cref{thm:w/l}.
\begin{theorem} \label{thm:w/l}
For $\varepsilon \geq 0$, and for all $n \in \mathbb{N}^* $, it holds that 
\begin{equation}
C_{\varepsilon, n}=\frac{1+e^{-\varepsilon}}{2} \coloneqq C_\varepsilon.
\end{equation}
\end{theorem}

\begin{proof}[Proof sketch (full proof in \Cref{app:w/l})]
We use a strategy similar to that of~\cite{hill1992survey}. 
We first prove \Cref{lem:reduc_w}, which shows that there exists a constant $\lambda$ such that replacing $X_1$ by $\lambda$ decreases the competitive ratio.
\begin{lemma}\label{lem:reduc_w}
For any independent random variables $X_1, \dots, X_n$ over $\RR_+$, let $\lambda = \Wn(X_{2:n})$. It holds that
\[
    \frac{\Wn(X_{1:n})}{\Ln(X_{1:n})} \geq \frac{\Wn(\lambda, X_{2:n})}{\Ln(\lambda, X_{2:n})}
\]
\end{lemma}
To prove \Cref{lem:reduc_w}, we show first that \[
\Wn(X_{1:n}) = \Wn(\lambda, X_{2:n}) + \max_{Q_1 \in \mathcal{Q}_\varepsilon} \EE[ \max(\EE[X_1 | Z_1] - \lambda, 0)].
\]
Then, we show that 
\[
\Ln(X_{1:n}) \leq \Ln(\lambda, X_{2:n}) + \max_{Q_1 \in \mathcal{Q}_\varepsilon} \EE[ \max(\EE[X_1 | Z_1] - \lambda, 0)].
\]
The second step is to prove \Cref{lem:long shot} which shows that replacing the last two variables by a long shot decreases the competitive ratio. 
\begin{lemma}
\label{lem:long shot}
For any random variables $X_1, \dots, X_n$ over $\RR_+$, let $\lambda = \Wn(X_{2:n})$.
For any $\alpha \in (0, \frac12)$, it holds that
\[
    \frac{\Wn(\lambda, X_{2:n})}{\Ln(\lambda, X_{2:n})} \geq \frac{\Wn(\lambda, X_{2:n-2}, L_\alpha)}{\Ln(\lambda, X_{2:n-2}, L_\alpha)}(1 - \alpha)
\]
where
\[
L_\alpha = \begin{cases}
    \frac{\Wn(X_{n-1}, X_n)}{\alpha} & \text{with probability $\alpha$} \\
    0 & \text{otherwise.}
\end{cases}
\]
\end{lemma}

The proof of this Lemma is in two parts.
In the first part, we show the lower bound 
\[
\Ln(\lambda, X_{2:n-2}, L_{\alpha}) \geq \max_{Q_{2:{n-2}}} \EE\bigg[\max\big(A, p (A +  \Wn(X_{n-1}, X_n)) \big) \bigg] - \alpha \Ln(\lambda, X_{2:n})
\]
where $A =\max_{i=2}^{n-2}\big(\lambda, \EE[X_i | Z_i]\big)$. To prove this inequality, the first and main argument is that the value of $L$ optimized over $Q_{n-1}$ is greater than the value obtained when $Q_{n-1}$ is set to a binary randomized response mechanism. The choice for $Q_{n-1}$ is motivated by the fact that $L_{\alpha}$ is binary. The resulting lower bound is then proven after steps that are detailed in \Cref{app:reduc_w}.

In the second part of the proof of \Cref{lem:long shot}, we establish that  
\[
     \Ln(\lambda, X_{2:n-2}, X_{n-1:n}) \leq \max_{Q_{2:n-2}} \EE[\max(A, p(A + \Wn(X_{n-1}, X_n))).
\]
The main technique is to combine Jensen's inequality and \Cref{prop:optimalchannel} multiple times to upper-bound the left-hand side of the above equation by a quantity that does not depend on $Q_{n-1}, Q_n$. The difficulty is that the optimal $Q_{n}$ should, in general, depend on $Q_{1:n-1}$ and no assumptions beyond finite expectation are made on $X_{n-1}, X_n$, so careful manipulations are required to keep the inequalities tight. The detailed proof is available in~\Cref{app:long shot}.

The worst-case instance is therefore with two variables, the first is a constant, and the second is a long shot. The proof is concluded by computing the competitive ratio on this instance.
\end{proof}

\Cref{thm:w/l} yields a tight closed-form expression of the competitive ratio $C_\varepsilon$,
which interpolates between the classical $1/2$ prophet guarantee when $\varepsilon\to\infty$ and the no-loss ratio $1$ when $\varepsilon=0$.
The key intuition is that LDP does not only harm the online decision maker: it also removes most of the prophet’s informational advantage.
Indeed, under LDP neither player ever observes the realizations $X_i$; both only see privatized reports, so even an offline agent cannot reliably identify the maximum ex post.
As $\varepsilon$ decreases, the reports become nearly uninformative and both agents are forced to rely mostly on priors, making their achievable values close. As $\varepsilon$ increases, the prophet gradually recovers an advantage, exactly quantified by $C_\varepsilon$.

% This theorem
% %is the main theorem of the paper. It 
% expresses $C_\varepsilon$ in closed form interpolating between the standard prophet inequality competitive ratio of $1/2$ when no privacy is required, $\varepsilon=\infty$, and the no loss ratio of $1$ when full privacy is required, $\varepsilon=0$. What this result means from a constrained optimization perspective is that adding privacy constraints to the online decision harms the online agent but reduces the expected utility relatively less quickly than that of the prophet. 

\begin{remark}
\label{remark:binary prophet}    
A consequence of \Cref{thm:w/l} is that the performance of $\Wn(\mathbf{X})$ relative to a private prophet that is \emph{restricted to binary mechanisms} is also given by $C_\varepsilon$ in the worst case. Specifically, let $\mathcal{B}_\varepsilon(\mathbf{X})$ denote the performance of an optimal \emph{binary} private prophet on instance $\mathbf{X}$. By \Cref{thm:w/l}, we have that 
$
\inf_{\mathbf{X} \in \mathcal{C}_n} \frac{\Wn(\mathbf{X})}{\mathcal{B}_\varepsilon(\mathbf{X})} \geq \inf_{\mathbf{X} \in \mathcal{C}_n} \frac{\Wn(\mathbf{X})}{\Ln(\mathbf{X})} = C_\varepsilon$.
On the other hand, since in the worst-case instance highlighted in \Cref{thm:w/l}, the private prophet uses a binary mechanism, it holds that 
\[
\inf_{\mathbf{X} \in \mathcal{C}_n} \frac{\Wn(\mathbf{X})}{\mathcal{B}_\varepsilon(\mathbf{X})} = C_\varepsilon.
\]
\end{remark}

\section{Conclusion and Future Work}
 
The study of online stopping under local differential privacy (LDP) addresses a fundamental tension in modern decision-making systems: balancing efficient online decisions with the protection of sensitive user data. Our main structural result, \Cref{th:optimal LDP stopping rule}, fully characterizes the optimal $\varepsilon$-LDP stopping strategy.
The optimal policy has a simple and interpretable \emph{privatize-then-decide} form: at each step $i$, the continuation value $W_{i+1}(\mathbf{X})$ plays the role of a threshold, the arriving agent releases a \emph{binary} report produced by randomized response applied to the event $\{X_i \ge W_{i+1}(\mathbf{X})\}$, and the decision maker stops at the first time the report indicates acceptance.
This yields an implementable algorithm with linear time complexity in the horizon $n$, and makes explicit how the privacy constraint compresses each arrival into a one-bit signal tailored to the stopping objective. A second contribution of this work is the sharp quantification of privacy-utility trade-offs via competitive ratios.
First, we show in \Cref{thm:cr_stopping_ldp} that the optimal $\varepsilon$-LDP online decision maker achieves the tight worst-case ratio
$
e^{\varepsilon}/(n-1+e^{\varepsilon})
$
compared to the best non-private online stopping rule, thereby characterizing the \emph{relative price of privacy}. 
Second, we compute in \Cref{thm:w/l} the worst-case competitive ratio against a privacy-preserving offline benchmark $C_\varepsilon \coloneqq (1+e^{-\varepsilon})/2$. This provides a meaningful target for private stopping algorithms: it isolates the residual advantage of offline adaptivity when both online and offline agents are constrained by the same $\epsilon$-LDP constraint.

% \Cref{th:optimal LDP stopping rule} provides an optimal online stopping algorithm based on a dynamically computed sequence of binary randomized response mechanisms combined with a simple decision rule. This result offers a solution to the problem of online stopping under local privacy that is practically relevant in settings where linear complexity in the number of variables~$n$ is acceptable.

% A key contribution of this work is the quantification of performance guarantees through competitive ratios.
% We first establish in~\Cref{thm:cr_stopping_ldp} that the worst-case competitive ratio of the $\varepsilon$-private optimal online decision-maker relative to the non-private optimal online decision-maker is
% $
% \frac{\exp(\varepsilon)}{\exp(\varepsilon)-1+n}.
% $
% This result quantifies, in the worst case, the fraction of utility lost as the privacy guarantee becomes stronger, thereby characterizing the relative price of privacy. The second major result of this study is the computation of the worst-case competitive ratio $C_\varepsilon$ of the $\varepsilon$-private optimal online decision-maker relative to the $\varepsilon$-private prophet. In~\Cref{thm:w/l}, we show that
% $
% C_\varepsilon = \frac{1+\exp(-\varepsilon)}{2}.
% $
% The competitive ratio relative to the private prophet thus provides a meaningful benchmark for comparing private optimal stopping algorithms.

Several open questions remain.
A first direction is to understand whether one can attain the private-prophet benchmark $C_\varepsilon$ with algorithms of substantially simpler form or reduced dependence on $n$, for instance, policies based on a fixed (or time-homogeneous) threshold or other constant-complexity algorithms. A second direction is to close the remaining gap in the characterization of the worst-case ratio relative to the \emph{non-private} prophet $\Mn$:
\Cref{prop:wagainstm} provides upper and lower bounds, and \Cref{prop: private T / M} gives an efficient threshold-based policy matching the lower bound, but the bounds do not coincide for all $(n,\varepsilon)$.
A third direction concerns the offline private benchmark itself: while \Cref{prop:staircase} shows that optimal staircase mechanisms exist, computing an optimal private prophet efficiently remains open; the binary-private-prophet variant introduced in \Cref{remark:binary prophet} may offer a more tractable surrogate without losing worst-case guarantees.
Finally, our results make precise that LDP can impose a substantial utility cost in the worst case, especially when $n$ is large.
It would be valuable to characterize analogous trade-offs under alternative privacy notions (e.g., $(\epsilon, \delta)$-DP~\cite{dwork2014algorithmic}, $\rho$-zCDP~\citep{ZeroDP}) and different trust models (e.g. central DP, shuffle DP~\citep{cheu2021differential}).

% Several open questions remain for future research. A central question arising from \Cref{th:optimal LDP stopping rule} and \Cref{thm:w/l} is whether there exists a locally private algorithm with complexity independent of~$n$—for example, one based on a fixed threshold—that achieves the competitive ratio~$C_\varepsilon$ relative to the private prophet. Another important direction is the characterization of the competitive ratio relative to the non-private prophet. \Cref{prop:wagainstm} provides upper and lower bounds, and \Cref{prop: private T / M} presents an efficient algorithm matching the lower bound; however, these bounds do not coincide for all values of~$\varepsilon$. A further challenge highlighted by this work is the computation of the optimal private prophet. \Cref{prop:staircase} shows that optimal staircase mechanisms exist, but whether they can be computed efficiently remains an open problem. The binary prophet introduced in \Cref{remark:binary prophet} may offer a more tractable alternative. Finally, while local differential privacy provides extremely strong data protection guarantees, \Cref{thm:cr_stopping_ldp} demonstrates that it comes at a high utility cost. Other privacy notions offer different privacy-utility trade-offs, which should be characterized precisely to understand when, how, and at what cost privacy can be protected.

% In the interest of anonymization, please do not include acknowledgements in your submission.
%
\begin{acks}
Achraf Azize and Vianney Perchet's research was supported in part by the French National Research Agency (ANR) in the framework of the PEPR IA FOUNDRY project (ANR-23-PEIA-0003) and through the grant DOOM ANR-23-CE23-0002. It was also funded by the European Union (ERC, Ocean, 101071601). Mathieu Molina has been funded by the European Research Council (ERC) under the European Union's Horizon Europe Program (FACT, grant agreement No.~101170373).
\end{acks}

% Bibliography
\bibliographystyle{ACM-Reference-Format}
\bibliography{bib}

\newpage 
% Appendix
\appendix
% \section{Supplementary materials}
\section{Missing proofs of \Cref{sec:optimalbackward}}
\subsection{Proof of \Cref{prop:two-points}}
\label{app:two-point}
\begin{proof}
Let $F$ be the distribution of $X$ and $Q$ be an $\varepsilon$-\ac{ldp} mechanism. The joint distribution of $(Z, X)$ is $\pi = Q \otimes F$.

We are going to construct $\tilde{\pi}$, such that $\pi_X=\tilde{\pi}_X = F$ i.e.  the marginal law over $X$ is the same, the conditional law $\tilde{\pi}_{Z \mid X}$ is $\varepsilon$-\ac{ldp}, and such that the support of $Z$ is only two points.\\

Let $\mathcal{Z}^+=\{ z \in \mathcal{Z} : \mathbb{E}[X \mid Z=z] \geq w\}$ and $\mathcal{Z}^-=\{ z \in  \mathcal{Z} : \mathbb{E}[X \mid Z=w] < w\}$ be the set of values of $Z$ where the posterior mean is respectively higher and lower than $w$. We consider the new modified message $\tilde{Z}=\mathds{1}[z \in \mathcal{Z}^+] \in \{0,1\}$, and let us denote by $\tilde{\pi}$ the induced law on $(X,\tilde{Z})$. We immediately have that the support is binary, and by the post-processing property, because $Z$ is $\varepsilon$-\ac{ldp} and $\tilde{Z}$ is a function of $Z$, that $\tilde{Z}$ is also $\varepsilon$-\ac{ldp}. It remains to show that modifying the message did not affect the value of $J_w$, so that $J_w(\pi)=J_w(\tilde{\pi})$.

First, we compute the value of $\pi$,
\begin{align*}
    J_w(\pi)&=\int_{\mathcal{Z}} \max(\mathbb{E}[X \mid  Z=z],w) d\pi_Z(z) \\
    &=\int_{\mathcal{Z^+}} \max(\mathbb{E}[X \mid  Z=z],w) d\pi_Z(z)+\int_{\mathcal{Z^-}} \max(\mathbb{E}[X \mid  Z=z],w) d\pi_Z(z)\\
    &=\int_{\mathcal{Z^+}} \mathbb{E}[X \mid  Z=z] d\pi_Z(z)+\int_{\mathcal{Z^-}} w d\pi_Z(z)\\ 
    &= \int_{\mathcal{Z}^+} \int_{\mathbb{R}_+} x  d\pi_{X \mid Z}(x) d\pi_Z(z) + w \Pr(Z \in \mathcal{Z}^-)\\
    &=\int_{\mathbb{R}_+} \int_{\mathcal{Z}^+}  x d \pi(x,z)+ w \Pr(Z \in \mathcal{Z}^-),
\end{align*}
where we apply Fubini-Tonelli for the last equality.

Then, the value of $\tilde{\pi}$ is, 
\begin{align*}
J_w(\tilde{\pi})&=\max(\mathbb{E}[X \mid \tilde{Z}=1],w) \Pr(\tilde{Z}=1) +\max(\mathbb{E}[X \mid \tilde{Z}=0],w)\Pr(\tilde{Z}=0) \\
&=\max(\mathbb{E}[X \mid z \in  \mathcal{Z}^+],w)\Pr(z \in  \mathcal{Z}^+) + \max(\mathbb{E}[X \mid z \in  \mathcal{Z}^-],w)\Pr(z \in  \mathcal{Z}^-) \\
&= \mathbb{E}[X \mid z \in  \mathcal{Z}^+] \Pr(z \in  \mathcal{Z}^+)+ w \Pr(z \in  \mathcal{Z}^-)\\
&= \frac{\mathbb{E}[X \mathds{1}[z \in  \mathcal{Z}^+]]}{\Pr(z \in  \mathcal{Z}^+)}\Pr(z \in  \mathcal{Z}^+)+w \Pr(z \in  \mathcal{Z}^-)\\
&= \int_{\mathbb{R}^+} \int_{\mathcal{Z}^+} x d\pi(x,z)+w\Pr(z \in  \mathcal{Z}^-) = J_w(\pi).
\end{align*}
which concludes the proof.
\end{proof}

\subsection{Proof of \Cref{prop:optimalchannel}}
\label{app:optimalchannel}

\begin{proof} By \Cref{prop:two-points}, we can restrict ourselves to two-point mechanisms, which are entirely described by the conditional law $\varphi(x)=\Pr(Z=1 \mid X=x)$ since $\Pr(Z=0 \mid X=x)=1-\varphi(x) $.\\

We do a case disjunction on domains of $\pi$, depending on how its posterior means relate to $w$. 
Suppose that $\pi$ is such that $\mathbb{E}[X \mid Z=1] \leq w$ and $\mathbb{E}[X \mid Z=0] \leq w$. Then $J_{w}=w$. Suppose now that $\mathbb{E}[X \mid Z=1] \geq w$ and $\mathbb{E}[X \mid Z=0] \geq w$. Then $J_{w}=\mathbb{E}[X]$.\\

Suppose now that $\mathbb{E}[X \mid Z=1] \geq w \geq \mathbb{E}[X \mid Z=0]$ (without loss of generality, we can flip the labels of $Z$). Denoting $F$ the law of $X$,

% For any joint distribution $\pi$, we have 
% \begin{equation*}
% J_{w}(\pi)=\mathbb{E}_{(X,Z) \sim \pi}[ \max(\mathbb{E}[X \mid Z],w)] \geq \max( \mathbb{E}[ \mathbb{E}[X \mid Z]],\mathbb{E}[w])=\max(\mathbb{E}[X],w).
% \end{equation*}
% % Suppose that $\pi$ is such that $\mathbb{E}[X \mid Z=1] \leq w$ and $\mathbb{E}[X \mid Z=0] \leq w$. 
% Then $J_{w}(\pi)=w \leq \max(\mathbb{E}[X],w)\leq J_{w}$, and is thus weakly sub-optimal. Similarly, suppose that both conditional expectations are greater than $w$. Then by the tower property of the conditional expectation, $J_{w}(\pi)=\mathbb{E}[\mathbb{E}[X \mid Z]]=\mathbb{E}[X] \leq \max(\mathbb{E}[X],w)\leq J_{w}$ so $\pi$ is also weakly suboptimal.\\ 

% Suppose now that we consider all $\pi$ such that $\mathbb{E}[X \mid Z=1] \geq w \geq \mathbb{E}[X \mid Z=0]$, or with labels of $Z$ flipped (they may not exist). 

% Therefore (without loss of generality by a simple relabelling), we focus on joint law $\pi$ such that $\mathbb{E}[X \mid Z=1] \geq w \geq \mathbb{E}[X \mid Z=0]$, assuming that at least one mechanism can induce such a $Z$. The case where no such mechanism exists is treated at the end.\\ 

\begin{align*}
J_{w}(\pi)&=\Pr(Z=1) \mathbb{E}[X \mid Z=1] +  w \cdot \Pr(Z=0)  \\
&= \left( \int_{\mathbb{R}_+} \varphi(x) dF(x) \right) \left(\int_{\mathbb{R}_+} x \frac{\varphi(x)}{\int_{\mathbb{R}_+} \varphi(x) dF(x)} dF(x) \right) + w \int_{\mathbb{R}_+} (1-\varphi(x)) dF(x)\\
&= w \int_{\mathbb{R}_+} dF(x) + \int_{\mathbb{R}_+} (x-w) \varphi(x) dF(x) \\
&=w + \int_{\mathbb{R}_+} (x-w) \varphi(x) dF(x).
\end{align*}
This is a linear objective in $\varphi$. We now drop the posterior mean condition, and simply maximize this linear program in $\varphi$. This yields a weakly larger objective than with the posterior mean constraint.

If the final optimal $\pi^*$ violates the posterior mean constraint, then either $J_{w}(\pi^*)=\mathbb{E}[X]$ or $J_{w}(\pi^*)=w$, which is weakly smaller than the values obtained in the previous cases. 
\begin{comment}
\hr{We should write this explicitely. Calling $J(\varphi) = w + \int_{\mathbb{R}_+} (x-w) \varphi(x) dF(x)$ we remark that (1) if $\varphi$ is such that $\EE[X | Z=0] \geq w$ and $\EE[X | Z=1] \geq w$, we have $J(\varphi) \leq \EE[X]$ and (2) of $\varphi$ is such that $\EE[X | Z=0] \leq w$ and $\EE[X | Z=1] \leq w$ then $J(\varphi) \leq w$. Therefore, denoting $\mathcal{M}$ the set of $\phi$ that satisfy $\varepsilon-LDP$ constraints, it holds that
\[
\max(\max_{\varphi \in \mathcal{M}} J(\varphi), \EE[X], w) = \max(\max_{\varphi \in \mathcal{M}: \EE[X | Z =1] \geq w \geq \EE[X | Z = 0]} J(\varphi), \EE[X], w).
\]
So discarding the constraint $\EE[X | Z =1] \geq w \geq \EE[X | Z = 0]$ when maximizing $J$ does not impact the original objective.
}
\end{comment}
Hence, using the value obtained by maximizing this linear program can never hurt the original objective.

% For any joint distribution $\pi$, we have 
% \begin{equation*}
% J_{w}(\pi)=\mathbb{E}_{(X,Z) \sim \pi}[ \max(\mathbb{E}[X \mid Z],w)] \geq \max( \mathbb{E}[ \mathbb{E}[X \mid Z]],\mathbb{E}[w])=\max(\mathbb{E}[X],w).
% \end{equation*}
% So, if the $\pi^*$ that maximizes the above linear program does not satisfy the posterior mean condition, then either $\mathbb{E}[X \mid Z=1]$ and $\mathbb{E}[X \mid Z=0]$ are both lower than $w$, then $J_{w}(\pi^*)=w$
% We will see that the resulting optimal $\varphi$ either satisfies the posterior mean condition, or it does not but simply does not contribute to the maximum.

The $\varepsilon$-\ac{ldp} constraints on $Q$ are equivalent to the following constraints on $\varphi$:
\begin{equation*}
\varphi(x)\leq e^{\varepsilon} \varphi(x') \quad \text{and} \quad (1-\varphi(x))\leq e^{\varepsilon} (1-\varphi(x')), \forall x,x'\in \mathbb{R}_+.
\end{equation*}
For $a= \inf \varphi$ and $b= \sup \varphi$, the privacy constraints are equivalent to the constraints $0\leq a \leq \varphi(x) \leq b \leq 1$, $ (1-a) \leq e^{\varepsilon}(1-b)$ and $b \leq e^{\varepsilon} a$. % \hr{show this}. 
Given that the objective is linear, this immediately implies for the optimal $\varphi$ that $\varphi(x)=a$ for $x< w$ and $\varphi(x)=b$ for $x>w$ (for $x=w$ this is indifferent). \\

The system of inequalities in $a$ and $b$ has $3$ extreme points, either $(a,b)=(0,0)$, $(a,b)=(1,1)$, or the intermediary extreme point
\begin{equation*}
a=\inf \varphi= \frac{1}{1+e^{\varepsilon}}, \quad b=\sup \varphi= \frac{e^{\varepsilon}}{1+e^{\varepsilon}}.
\end{equation*}
So the maximum is reached for one of these three points, where we remark that choosing an extreme point with either $a=b=0$ or $a=b=1$ recovers the non-informative mechanisms $\mathbb{E}[X \mid Z]=\mathbb{E}[X]$, yielding a value of $J_{w}(\pi)=\max(\mathbb{E}[X],w)$ obtained previously. The intermediary mechanism corresponds to the binary randomized response  $Z \sim \mathrm{Ber}(\frac{e^{\varepsilon}}{1+e^{\varepsilon}}\mathds{1}[X \geq w]+ \frac{1}{1+e^{\varepsilon}} \mathds{1}[X<w])$.\\

From there, letting $p=e^{\varepsilon}/(1+e^{\varepsilon})$,  assuming that $\mathbb{E}[X \mid Z=1] \geq w \geq \mathbb{E}[X \mid Z=0]$ 
%\hr{we don't need this assumption since we already said that we can maximize $J(\phi)$ without constraints}
(once again, this cannot hurt the objective), using $(x-w)^+=\max(x,w)-w$ and $(w-x)^+=w-\min(x,w)$, 
\begin{align*}
J_{w}(\pi^*)&=w+\mathbb{E}[ p (X-w)  \mathds{1}[X \geq w] + (1-p)(X-w) \mathds{1}[X<w]]\\
&=w+ p \mathbb{E}[(X-w)^+]-(1-p)\mathbb{E}[(w-X)^+]\\
&=w+p\mathbb{E}[\max(X,w)]-pw -( (1-p) (w - \mathbb{E}[\min(X,w)]))\\
&=p \mathbb{E}[\max(X,w)] + (1-p) \mathbb{E}[\min(X,w)]\\
&=\mathbb{E}[\phi(X,w)].
\end{align*}

Overall, this binary randomized response mechanism is optimal if and only if (up to ties) $\mathbb{E}[\phi(X,w)] \geq \max(\mathbb{E}[X],w)$. Through backward programming for partial observations \citep{Assaf1998ASV}, we select $X$ whenever $\mathbb{E}[X\mid Z] \geq w$ which happens if and only if $Z=1$. Otherwise, the indifferent mechanism, i.e. a constant message, is optimal. For convenience, we can simply choose a constant message $Z=1$ when $\mathbb{E}[X]>w$, so that he message $Z$ being equal to $1$ is equivalent to accepting $X$ in all situations.
\end{proof}

\subsection{Proof of \Cref{prop:interactive}}
\label{app:interactive}
\begin{proof}

Fix $(Q_i)_{i \in [n]}$ a set of interactive privacy mechanisms and $(Z_i)_{i \in [n]}$ the associated random variables.  Call $Y_i$ the expected value given by by the optimal decision rule given $Z_{1:i}$ for any $i \in [n+1]$. We have $Y_{n+1} = 0$. For $i \in [n]$, either the optimal decision is to stop at time $i$ in which case $Y_i = \EE[X_i \mid Z_{1:i}]$ otherwise, $Y_i = \EE[ Y_{i+1} \mid Z_{1:i}]$.
Therefore $(N_i= \EE[Y_i \mid Z_{1:i-1}])_{i \ in [n]}$ can be computed recursively from the recursion
$N_{n+1} = 0$ and for $1 \leq k \leq n$, 
\[
N_{k} =\EE[\max(\EE[X_k | Z_1, \dots, Z_k], N_{k+1}) | Z_1, \dots, Z_{k-1}].
\]
The expected value attained by the optimal stopping rule is given by $N_1$.
Then, consider $(W_k)_{k \in [n]}$ as defined in \Cref{th:optimal LDP stopping rule}. We show by induction that $N_k \leq W_k$ for any $1 \leq k \leq n+1$.
First, $N_{n+1} = W_{n+1} = 0$. Then for $1 \leq k \leq n$, assume $N_{k+1} \leq W_{k+1}$ holds, and denote $\EE_{k-1}[ \cdot ] = \EE[ \cdot \mid Z_1, \dots, Z_{k-1}]$.
We can write
\begin{align*}
N_k &= \EE_{k-1}[\max(\EE_{k-1}[X_k | Z_k], N_{k+1})] \\
&\leq \EE_{k-1}[\max(\EE_{k-1}[X_k | Z_k], W_{k+1})] && \text{(by recursion hypothesis)} \\
&= \EE_{k-1}[\max(\EE[X_k | Z_k], W_{k+1})] && \text{($X_k$ and $Z_{1:k-1}$ independent)} \\
&\leq \max(\EE_{k-1}[X_k], \EE_{k-1}[\phi(X_k, W_{k+1})], W_{k+1}) && \text{(By \Cref{prop:optimalchannel})} \\
&= \max(\EE[X_k], \EE[\phi(X_k, W_{k+1})], W_{k+1}) && \text{(since $X_k$ and $Z_1, \dots, Z_{k-1}$ independent)} \\
&= W_k
\end{align*}
which concludes the proof since $W_1$ is reached by the optimal non-interactive mechanism.

\end{proof}

\section{Missing proofs of \Cref{sec:competitive}}
\subsection{Proof of \Cref{prop:lb-w/v}}
\label{app:lb-w/v}
\begin{proof}

We begin the proof by showing a technical Lemma.
\begin{lemma}\label{lem:pointwise_bound}
Let $p\in(0,1)$ and $r\in[0,1]$. For
\begin{equation*}
\alpha \coloneqq \frac{p\,r}{\,p+(1-p)r\,},
\qquad
\lambda \coloneqq 1-\frac{\alpha}{p}
= 1-\frac{r}{\,p+(1-p)r\,},
\end{equation*}
we have for all $y\ge 0$ the following inequality,
\begin{equation}\label{eq:pm-ineq}
\lambda\,r \;+\; (1-\lambda)\Bigl(p\max\{y,r\}+(1-p)\min\{y,r\}\Bigr)
\;\ge\; \alpha\,\max\{y,1\}.
\end{equation}
\end{lemma}

\begin{proof}[Proof of \Cref{lem:pointwise_bound}]
Note that $0\le \alpha\le r$ (since the denominator in $\alpha$ is at least $p$), and $\alpha \leq p$ because $r\leq 1$, and thus $0\leq\lambda\leq1$.
Define
\begin{equation*}
 L(y)\coloneqq \lambda r + (1-\lambda)\bigl(p\max\{y,r\}+(1-p)\min\{y,r\}\bigr),
\qquad
R(y)\coloneqq \alpha\,\max\{y,1\}.   
\end{equation*}
We will show $L(y)\ge R(y)$ case by case.

\smallskip\noindent
\emph{Case 1: $0\le y\le r\le 1$.}
Here $\max\{y,r\}=r$, $\min\{y,r\}=y$, and $\max\{y,1\}=1$, so
\begin{equation*}
L(y)=\lambda r+(1-\lambda)\bigl(pr+(1-p)y\bigr)
= r\bigl(\lambda+p(1-\lambda)\bigr) + (1-\lambda)(1-p)\,y.
\end{equation*}
This is affine and nondecreasing in $y$, hence minimized at $y=0$:
\begin{align*}
L(0)&= r\bigl(\lambda+p(1-\lambda)\bigr)
= r\Bigl(1-(1-\lambda)(1-p)\Bigr)
= r\Bigl(1-\frac{(1-p)r}{\,p+(1-p)r\,}\Bigr)\\
&= \frac{pr}{\,p+(1-p)r\,}
= \alpha
= R(y).
\end{align*}
Thus $L(y)\ge R(y)$ on $[0,r]$.

\smallskip\noindent
\emph{Case 2: $r\le y\le 1$.}
Now $\max\{y,r\}=y$, $\min\{y,r\}=r$, and $\max\{y,1\}=1$, hence
\begin{equation*}
L(y)=\lambda r + (1-\lambda)\bigl(py+(1-p)r\bigr)
= (1-\lambda)p\,y + r\bigl(\lambda+(1-\lambda)(1-p)\bigr).
\end{equation*}
This is non-decreasing in $y$, so its minimum on $[r,1]$ occurs at $y=r$:
\begin{equation*}
L(r)=\lambda r + (1-\lambda)\bigl(pr+(1-p)r\bigr) = r.
\end{equation*}
Since $\displaystyle \alpha=\frac{pr}{p+(1-p)r}\le r$, we have $L(y)\ge r\ge \alpha=R(y)$.

\smallskip\noindent
\emph{Case 3: $y\ge 1$.}
Again $\max\{y,r\}=y$, $\min\{y,r\}=r$, while $\max\{y,1\}=y$. Using $\lambda=1-\alpha/p$,
\begin{equation*}
L(y)=(1-\lambda)p\,y + r\bigl(\lambda+(1-\lambda)(1-p)\bigr)
= \alpha\,y + r\bigl(\lambda+(1-\lambda)(1-p)\bigr)
\ge \alpha\,y
= R(y).
\end{equation*}

All cases give $L(y)\ge R(y)$, proving \eqref{eq:pm-ineq}.
\end{proof}

We now prove \Cref{prop:lb-w/v}.
For $i \in [n]$, let $Y=X_i/V_{i+1}$, $r_{i+1}=W_{i+1}/V_{i+1}\leq1$, $\alpha = \frac{pr_{i+1}}{p + (1 - p) r_{i+1}}$ where $p = \frac{e^{\varepsilon}}{1 + e^{\varepsilon}}$ and fix $\lambda \in [0,1]$. From the recursion~\Cref{eq:wphi}, it holds that
\begin{align*}
W_i&=\max\left(W_{i+1},\mathbb{E}[X_i],p\mathbb{E}[\max(X_i,W_{i+1})]+(1-p)\mathbb{E}[\min(X_i,W_{i+1})]\right) \\
&\geq \max(W_{i+1},p\mathbb{E}[\max(X_i,W_{i+1})]+(1-p)\mathbb{E}[\min(X_i,W_{i+1})])\\
&\overset{(1)}{\geq} \lambda W_{i+1}+(1-\lambda)\left( p\mathbb{E}[\max(X_i,W_{i+1})]+(1-p)\mathbb{E}[\min(X_i,W_{i+1})]\right)\\
&=V_{i+1}\mathbb{E}[\lambda r_{i+1}+(1-\lambda)\left( p\max(Y,r_{i+1})+(1-p)\min(Y,r_{i+1})\right)]\\
&\overset{(2)}{\geq} V_{i+1} \alpha \mathbb{E}[\max(Y,1)]\\
&= \alpha \mathbb{E}[\max(X_i,V_{i+1})]\\
&= \frac{p r_{i+1}}{p+(1-p)r_{i+1}} V_i,
\end{align*}
where $(1)$ uses that the maximum of two points is above any convex combination, $(2)$ is by \Cref{lem:pointwise_bound} applied pointwise and then taking the expectation.\\

Therefore, $r_i \geq (pr_{i+1})/(p+(1-p)r_{i+1})$ or in other words $1/r_i \leq 1/r_{i+1}+(1-p)/p$. By plugging-in the definition of $r_i$, we get
\begin{equation}
\label{eq:w/v}
\frac{V_i}{W_i} \leq \frac{V_{i+1}}{W_{i+1}} + \frac{1-p}{p} 
\end{equation}
Iterating this inequality with the initialization $\frac{V_n}{W_n}=1$, yields the desired inequality 
\[
\frac{V_1}{W_1} \leq 1+ (n-1) (1-p)/p=1+(n-1)e^{-\varepsilon}.
\]
\end{proof}

\label{app:competitive}

\subsection{Proof of \Cref{prop:ub-w/v}}
\label{app:ub-w/v}
\begin{proof}
We first compute $\mathbb{E}[X_i], \mathbb{E}[X_i \mathds{1}[ Z_i=1]]$ and $\mathbb{E}[X_i \mathds{1}[Z_i=0]]$ for $R_{W_{i+1},\varepsilon}$ the optimal mechanism at step $i$ with respect to a future value of $W_{i+1}$, denoting $p=e^{\varepsilon}/(1+e^{\varepsilon})$.

\begin{align*}
\mathbb{E}[X_i]&=q+e^{-\varepsilon}(1-q) \leq q+(1-q)=1,\\
\mathbb{E}[X_i \mathds{1}[Z_i=1]]&=\mathbb{E}[X_i \Pr(Z_i=1 \mid X_i)]= (1+e^{-\varepsilon}\frac{1-q}{q}) \cdot \frac{e^{\varepsilon}}{1+e^{\varepsilon} } \cdot q= p \mathbb{E}[X_i]\\
\mathbb{E}[X_i \mathds{1}[Z_i=0]]&=(1+e^{-\varepsilon}\frac{1-q}{q}) \cdot \frac{1}{1+e^{\varepsilon} } \cdot q= (1-p) \mathbb{E}[X_i]
\end{align*}

We will now show the following property by recursion: $W_i=1$ for all $i \in [n]$. We already have $W_n=\mathbb{E}[X_n]=1$. Now assuming that $W_{i+1}=1$ and using the recursion of \Cref{th:optimal LDP stopping rule}, we will prove this holds for $W_i$. 

We can compute $W_i$ as the maximum of three terms. Because $\mathbb{E}[X_i]\leq 1 = W_{i+1}$, we only need to compare $W_{i+1}$ with the last term. 
We first compute each part:
\begin{align*}
&\mathbb{E}[\max(\mathbb{E}[X_i \mid Z_i],W_{i+1})]= \Pr(Z_i=1) \max(\mathbb{E}[X_i \mid Z_i=1],1)+\Pr(Z_i=0) \max(\mathbb{E}[X_i \mid Z_i=0],1)\\
&= \max(\mathbb{E}[X_i \mathds{1}[Z_i=1]],\Pr(Z_i=1))+ \max(\mathbb{E}[X_i \mathds{1}[Z_i=0]],\Pr(Z_i=0))\\
&=\max(p\mathbb{E}[X_i],pq+(1-p)(1-q))+\max((1-p)\mathbb{E}[X_i],p(1-q)+q(1-p)).
\end{align*}
But, $p \mathbb{E}[X_i]= pq + (1-p)(1-q)$, and $(1 - p) \mathbb{E}[X_i] = (1-p)q + (1-p)^2/p(1-q) \leq p(1-q)+q(1-p)$ given that $(1-p)^2/p \leq p$ as $p \geq 1/2$. Therefore

\begin{equation*}
 \mathbb{E}[\max(\mathbb{E}[X_i \mid Z_i],W_{i+1})]\leq  \Pr(Z_i=1)+\Pr(Z_i=0)=1.
\end{equation*}
Overall,
\begin{equation*}
p\mathbb{E}[\max(\mathbb{E}[X_i \mid Z_i],W_{i+1})]+ (1-p) \mathbb{E}[\min(\mathbb{E}[X_i \mid Z_i],W_{i+1})] \leq p\cdot 1+ (1-p) \cdot 1=1
\end{equation*}

Hence, we directly have the desired recurrence 
\begin{align*}
W_i&=\max(\mathbb{E}[X_i],W_{i+1},p\mathbb{E}[\max(\mathbb{E}[X_i \mid Z_i],W_{i+1})]+ (1-p)\mathbb{E}[\min(\mathbb{E}[X_i \mid Z_i],W_{i+1})])\\
&=W_{i+1}=1.
\end{align*}

In particular, $W_{\varepsilon}^{(n)}(\mathbf{X}^q)=W_1=1$. \\

Let us now compute the value of the standard, backward dynamic programming $V_i$. We have, for $q$ small enough such that $1+e^{-\varepsilon}(1-q)/q \geq V_{i+1}$, then
\begin{comment}
\hr{$V_{i+1}$ also depends on $q$ so the argument does not hold directly.
Fix $q \in (0, 1)$, it holds that $\forall i \in [n], V_i \leq 1 + \exp(-\varepsilon)(1 - q) / q$.
Indeed, $V_n = 1 \leq 1 + \exp(-\varepsilon)(1 - q) / q$ and $V_{i+1} \leq 1 + \exp(-\varepsilon)(1 - q) / q$ implies $V_i = \EE[\max(X_i, V_{i+1})] \leq 1 + \exp(-\varepsilon)(1 - q) / q$.
}
\end{comment}
\begin{equation*}
V_i=\mathbb{E}[\max(X_i,V_{i+1})]=(q+e^{-\varepsilon}(1-q))+(1-q)V_{i+1}.
\end{equation*}
Using that $V_n=1$, and taking $q \rightarrow 0$, we have $V_1= (n-1)e^{-\varepsilon}+1$.

This yields $W_1/V_1=1/((n-1)e^{-\varepsilon}+1)$ which is exactly the desired inequality.
\end{proof}

\subsection{Proof of \Cref{cor:distinct privacy}}
\label{app:distinct privacy}
\begin{proof}
    The lower bound is proved by adapting the proof of \Cref{prop:lb-w/v}. We can obtain the inequality 
    \[
    \frac{V_i}{W_i} \leq \frac{V_{i+1}}{W_{i+1}} + \frac{1 - p_i}{p_i}
    \]
    where $p_i = \frac{\exp(\varepsilon_i)}{1 + \exp(\varepsilon_i)}$, by replacing $p = \frac{\exp(\varepsilon)}{1 + \exp(\varepsilon)}$  in the proof of \Cref{eq:w/v} by $p_i$ and applying the same steps.
    Then iterating and using that $V_n / W_n = 1$, we get
    \[
    \frac{V_1}{W_1} \leq 1 + \sum_{i=1}^{n-1} \exp(-\varepsilon_i).
    \]

    The upper bound is proved by adapting the proof of \Cref{prop:ub-w/v}. Consider the following sequence: $X^q_n=1$ a.s. and for $i \in [n-1]$,
\begin{equation*}
X^q_i= \begin{cases}
    1+e^{-\varepsilon_i}\frac{1-q}{q}, \quad \text{with probability $q$},\\
    0, \quad \text{with probability $1-q$}.
\end{cases}
\end{equation*}
The same computation apply, replacing $p$ by $p_i$ when computing $\EE[X_i], \EE[X_i \ind{Z_i=1}], \EE[X_i \ind{Z_i=0}]$ and when using the recursion of \Cref{th:optimal LDP stopping rule}. We get $W_i = 1$ for any $i \in [n]$ and 
\begin{equation*}
V_i=\mathbb{E}[\max(X_i,V_{i+1})]=(q+e^{-\varepsilon_i}(1-q))+(1-q)V_{i+1}.
\end{equation*}
Using $V_n = 1$ and taking $q \rightarrow 0$, we get $V_1 =  1 + \sum_{i=1}^{n-1} \exp(-\varepsilon_i)$ and therefore
\[
\frac{W_1}{V_1} = \frac1{1 + \sum_{i=1}^{n-1} \exp(-\varepsilon_i)}.
\]

\end{proof}

\subsection{Proof of \Cref{prop:wagainstm}}
\label{app:wagainstm}
\begin{proof}
First note that by using the $0$-\ac{ldp} deterministic stopping rule $\tau=\argmax_{i \in [n]} \mathbb{E}[X_i]$, as $\mathbb{E}[\max_{i \in [n]} X_i] \leq \sum_{i \in [n]} \mathbb{E}[X_i]\leq n\max_{i \in [n]}\mathbb{E}[X_i]$ we can always achieve at least a $1/n$ competitive ratio against the non-private prophet. Hence, the lower bound in $\frac1{n}$. 
The lower bound $\frac{e^{\varepsilon}}{2e^{\varepsilon} - 2 + 2n}$ is immediate from \Cref{thm:cr_stopping_ldp} and the classical prophet lower bound $V/\mathbb{E}[\max_i X_i]\geq 1/2$, by writing out $W/\mathbb{E}[\max_i X_i]=(W/V)\cdot (V/\mathbb{E}[\max_i X_i])$.

For the upper bound, we will use a similar counter-example as in \Cref{thm:cr_stopping_ldp}. Use $X_i \sim X^q$ for $i=2,\dots,n-1$, $X_n=1/c$ with probability $c \in [0,1]$ and $X_n=0$ with probability $1-c$, and $X_1=1$ a.s.. Because $\mathbb{E}[X_n]=1$, the same reasonning can be applied for $W_1$ and $q$ small enough. Similarly, adding the constant of $1$ in front of the sequence does not increase the value of the optimal \ac{ldp} stopping, hence $W_1=1$.

We then take $c$ small enough so that $1/c> 1+e^{-\varepsilon}(1-q)/q$. There are three distincts relevant events: either $X_n=1/c$ so that $\max_i X_i=1/c$ which happens with probability $c$, $X_n=0$ and one of the $X^q$ is non-zero which happens with probability $(1-c)(1-(1-q)^{n-2})$, or all $X_i=0$ for $i>1$ which happens with probability $(1-c)(1-q)^{n-2}$. Overall,
\begin{align*}
\mathbb{E}[\max_i X_i]&=c \cdot \frac{1}{c}+(1-c)(1-(1-q)^{n-2})\cdot (1+e^{-\varepsilon}(\frac{1}{q}-1))+ (1-c)(1-q)^{n-2} \cdot 1 \\
&\xrightarrow[c \rightarrow 0]{} 1+(1-(1-q)^{n-2}) \cdot (1+e^{-\varepsilon}(\frac{1}{q}-1))+(1-q)^{n-2} \\
&=1+((n-2)q+o(q)) \cdot (1+e^{-\varepsilon}(\frac{1}{q}-1))+1-(n-2)q+o(q) \\
&= 2 + e^{-\varepsilon} (n-2)+  (1-e^{-\varepsilon})(n-2)q -(n-2)q +o(q)+o(1) \\
& \xrightarrow[q \rightarrow 0]{} 2+e^{-\varepsilon}(n-2),
\end{align*}
which immediately yields the desired upper bound. 
\end{proof}

\subsection{Proof of \Cref{prop: private T / M}}
\label{app:private T / M}

\begin{proof}
Let $i \leq n - 1$, $\alpha_i=\frac{1}{n-i+e^{\varepsilon}}$ and $\beta_i=\frac{e^{\varepsilon}}{n-i+e^{\varepsilon}}$. We have that 
$\frac{\beta_i}{\alpha_i} = e^\varepsilon$, and 
\begin{align*}
    \frac{1 - \alpha_i}{1 - \beta_i} &= 1 + \frac{e^\varepsilon - 1}{n - i}\\
    & \leq 1 + (e^\varepsilon - 1) = e^\varepsilon
\end{align*}
since $n - i \geq 1$. This means that $Q_i \in \mech_\varepsilon$.

Then we have 
\begin{align*}
\mathbb{E}[X_{\tau_{\varepsilon}}]&=\sum_{i \in [n]} \mathbb{E}[X_i \mathds{1}[\tau_{\varepsilon}=i]]\\ 
&=\sum_{i \in [n]} \mathbb{E}[X_i \mathbb{E}[\mathds{1}[\tau_{\varepsilon}=i] \mid X_1,\dots,X_n]].
\end{align*}

We then lower bound $\mathbb{E}[\mathds{1}[\tau_{\varepsilon}=i] \mid X_1,\dots,X_n]$. It holds that
\begin{align*}
&\mathbb{E}[\mathds{1}[\tau_{\varepsilon}=i] \mid X_1, \dots, X_n]\\
&=\mathbb{E}\bigg[\mathds{1}\bigg[Z_i = 1 \wedge \bigwedge_{j=1}^{i-1} Z_j = 0\bigg] \mid X_1, \dots, X_n\bigg]\\
&=\Pr(Z_i=1 \mid X_1,\dots,X_n)\prod_{j=1}^{i-1} \Pr(Z_j=0 \mid X_1,\dots,X_n)\\
&=\Pr(Z_i=1 \mid X_i)\prod_{j=1}^{i-1} \Pr(Z_j=0 \mid X_j)\\
&=(\mathds{1}[X_i\geq T] \beta_i + \mathds{1}[X_i <T] \alpha_i) \prod_{j=1}^{i-1} \left( \mathds{1}[X_j\geq T] (1-\beta_j) + \mathds{1}[X_j <T] (1-\alpha_j) \right)\\
&\geq \mathds{1}[X_i\geq T] \beta_i \prod_{j=1}^{i-1} \mathds{1}[X_j <T] (1-\alpha_j) \\
&= \left( \beta_i \prod_{j=1}^{i-1} (1-\alpha_j) \right) \left(  \mathds{1}[X_i \geq T] \prod_{j=1}^{i-1} \mathds{1}[X_j < T] \right) \\
& =\left( \beta_i \prod_{j=1}^{i-1} (1-\alpha_j) \right) \mathds{1}[\tau(T)=i]. 
\end{align*}

The important part of the stopping rule, the mechanisms $Q_i$, were chosen in order to equalize the product of the $\alpha_i$ and $\beta_i$, which is equal to 
\begin{equation*}
 \beta_i \prod_{j=1}^{i-1} (1-\alpha_j) = \frac{e^{\varepsilon}}{n-i+e^{\varepsilon}} \prod_{j=1}^{i-1} \left(1- \frac{1}{n-j+e^{\varepsilon}}\right)=\frac{e^{\varepsilon}}{n-i+e^{\varepsilon}} \prod_{j=1}^{i-1} \frac{n-(j+1)+e^{\varepsilon}}{n-j+e^{\varepsilon}}= \frac{e^{\varepsilon}}{n-1+e^{\varepsilon}}.
\end{equation*}
Hence, 
\begin{align*}
\mathbb{E}[X_{\tau_{\varepsilon}}]&=\sum_{i \in [n]} \mathbb{E}[X_i \mathds{1}[\tau_{\varepsilon}=i]]\\ 
&=\sum_{i \in [n]} \mathbb{E}[X_i \mathbb{E}[\mathds{1}[\tau_{\varepsilon}=i] \mid X_1,\dots,X_n]] \\
& \geq \sum_{i \in [n]} \mathbb{E}[X_i  \frac{e^{\varepsilon}}{n-1+e^{\varepsilon}}\mathds{1}[\tau(T)=i ]]\\
&= \frac{e^{\varepsilon}}{n-1+e^{\varepsilon}} \mathbb{E}[X_{\tau(T)}]. \qedhere
\end{align*}
\end{proof}

\section{Missing proofs of \Cref{sec:competitive ratio}}

\subsection{Proof of \Cref{lemma:sublinear}}
\label{app:sublinear}
\begin{proof}
Let $V\coloneqq (v_1,\dots,v_m)$, for $i \in [n]$ denote by $P^{(i)}=(\Pr(X_i=v_1),\dots,\Pr(X_i=v_m))$ the probability vector of $X_i$, and define $U^{(i)}=(v_1 \Pr(X_i=v_1),\dots,v_m \Pr(X_i=v_m))$ the coordinate wise product between $P^{(i)}$ and $V$. 
Let $Q_1,\dots,Q_n \in \mathcal{Q}_{\varepsilon}^{n}$ and $\pi^{i}$ be the corresponding joint law on $(X_i,Z_i)$. For $z \in \mathcal{Z}$, we define by $Q_{i,z}=(Q_i(z \mid v_1),\dots,Q_i(z \mid v_m))$. For $j \in [n]$, we can explicit the posterior mean and the probability of a given $Z_i$:
\begin{align*}
&\Pr(Z_i=z)=\sum_{j=1}^m \Pr(Z_i=z \mid X_i=v_j) \Pr(X_i=v_j)= \langle Q_{i,z} ,P^{i}  \rangle,\\
& \mathbb{E}[X_i \mid Z_i=z]= \sum_{j=1}^m v_j \frac{\Pr(X_i=v_j) \Pr(Z_i=z \mid X_i=v_j)}{\Pr(Z_i=z)}=\frac{\langle Q_{i,z}, U^{(i)}\rangle}{\Pr(Z_i=z)}.
\end{align*}
Using the above formulas, we have
\begin{align*}
L(Q_1,\dots,Q_n)&=\sum_{z_1,\dots,z_n} \left( \prod_{i=1}^{n} \Pr(Z_i=z_i)\right) \max_{i \in [n]} \mathbb{E}[X_i \mid Z_i=z_i]  \\
&=\sum_{z_1,\dots,z_n} \left( \prod_{i=1}^{n} \Pr(Z_i=z_i)\right) \max_{i \in [n]} \frac{\langle Q_{i,z_i}, U^{(i)}\rangle}{\Pr(Z_i=z_i)} \\
&= \sum_{z_1,\dots,z_n} \max_{i \in [n]} \left \{ \langle Q_{i,z_i}, U^{(i)}\rangle  \prod_{j \in [n]\setminus \{i\}} ^{n} \Pr(Z_j=z_j) \right\}\\
&= \sum_{z_1,\dots,z_n} \max_{i \in [n]} \left \{ \langle Q_{i,z_i}, U^{(i)}\rangle  \prod_{j \in [n]\setminus \{i\}} ^{n} \langle Q_{j,z_j}, P^{(j)} \rangle \right\}.
\end{align*}
Now, fixing all other $Q_j$ for $j\neq i$, and fixing a single vector $(z_1,\dots,z_n)$, the maximum term is $\max_j( \langle Q_{i,z}, t_j\rangle  c_j)$ for some vector $t_j$ and coefficient $c_j$. It is homogenous in $Q_{i,z}$ by linearity of the scalar product, and it is sub-additive as 
\begin{equation*}
\max((x_1+x_2)y_1,(x_1+x_2)y_2) =\max(x_1y_1+x_2y_1,x_1y_2+x_2y_2) \leq  \max(x_1y_1,x_1y_2)+\max(x_2y_1,x_2y_2).
\end{equation*}
A sum of homogeneous and sub-additive functions is still homogeneous and sub-additive, hence so is $Q_i \mapsto L(Q_1,\dots,Q_n)$.
\end{proof}

\subsection{Proof of \Cref{thm:w/l}}
\label{app:w/l}
\begin{proof}
We use a strategy similar to that of~\cite{hill1992survey}. In \Cref{lem:reduc_w}, we show that there exists a constant $\lambda$ such that replacing $X_1$ by $\lambda$ decreases the competitive ratio, and in \Cref{lem:long shot}, that replacing the last two variables by a long-shot also decreases the competitive ratio. The worst-case instance is therefore with two variables, the first is a constant, and the second is a long shot. The proof concludes by computing the competitive ratio on this instance. The proof of \Cref{lem:reduc_w} and \Cref{lem:long shot} is given in \Cref{app:long shot} and \Cref{app:reduc_w}.

We now use the two previous Lemmas to finish the proof of \Cref{thm:w/l}. 
Fix any random variables $X_{1:n}$ over $\RR_+$, let $\lambda = \Wn(X_{2:n})$.

Then,
\begin{align*}
    \frac{\Wn(X_{1:n})}{\Ln(X_{1:n})} &\geq \frac{\Wn(\lambda, X_{2:n})}{\Ln(\lambda, X_{2:n})} && \text{(By \Cref{lem:reduc_w})}
\end{align*}

When $n=1$, the competitive ratio is $\Wn(X_1) / \Ln(X_1) = \frac{\EE[X_1]}{\EE[X_1]} = 1$. We consider separately the cases $n=2$ and $n > 2$.

\paragraph{Case $n=2$} If $n=2$, $\lambda = \EE[X_2]$

\begin{align*}
    \Wn(\lambda, X_2) = \EE[X_2]
\end{align*}
and
\begin{align*}
    \Ln(\lambda, X_2) &= \max_{Q_2} \EE[\max(\EE[X_2], \EE[X_2 \mid Z_2])] \\
    &= \max(\EE[X_2], \EE[\phi(X_2, \EE[X_2])]) \\
    &\leq 2p(\EE[X_2]).
\end{align*}
So that for $n=2$
\[
\frac{\Wn(X_{1:n})}{\Ln(X_{1:n})} \geq \frac1{2p}.
\]

\paragraph{Case $n > 2$}
Fix $\alpha < \frac12$, it holds
\begin{align*}
    \frac{\Wn(\lambda, X_{2:n})}{\Ln(\lambda, X_{2:n})} &\geq \frac{\Wn(\lambda, X_{2:n-2}, L_\alpha)}{\Ln(\lambda, X_{2:n-2}, L_\alpha)}(1 - \alpha) && \text{(By \Cref{lem:long shot})} \\
    &\geq \frac{\Wn(\lambda, L_\alpha)}{\Ln(\lambda, L_\alpha)}(1 - \alpha)^{n-2} && \text{(By \Cref{lem:long shot})} \\
    &\geq \frac{1}{2p}(1 - \alpha)^{n-2}
\end{align*}
Taking the limit $\alpha \rightarrow 0$ gives 
\[
\forall n \geq 1, \frac{\Wn(X_{1:n})}{\Ln(X_{1:n})} \geq \frac{1}{2p}.
\]
Then, taking the infimum on the left we obtain
\[
\inf_{\Dcal_1, \dots, \Dcal_n \in \Delta(\RR_+)} \frac{\Wn(X_{1:n})}{\Ln(X_{1:n})} \geq \frac1{2} \frac{\exp(\varepsilon) + 1}{\exp(\varepsilon)}
\]

\paragraph{Equality case} For the equality case, take $n=2$ and consider $X_1 = 1$ and for $\alpha < 1$
\begin{align*}
    X_2 = \begin{cases}
        \frac1{\alpha} & \text{with probability $\alpha$} \\
        0 & \text{with probability $1 - \alpha$}
    \end{cases}
\end{align*}
We have following the case $n=2$ studied above $\Wn(X_1, X_2) = 1$ and
\begin{align*}
    \Ln(X_1, X_2) &= \EE[\phi(X_2, \EE[X_2])] \\
    &= \alpha (p \frac1{\alpha} + (1 - p)) + (1 - \alpha) p \\
    &= p + (1 - p)\alpha + (1 - \alpha)p \\
    &= 2p - (2p - 1) \alpha
\end{align*} 
Therefore
\begin{align*}
\inf_{\alpha \in (0, \frac12)} \frac{\Wn(X_1, X_2)}{\Ln(X_1, X_2)} &= \inf_{\alpha \in (0, \frac12)} \frac{1}{2p - (2p - 1) \alpha} \\
&= \frac1{2p}
\end{align*}
which concludes the proof.
\end{proof}

\subsubsection{Proof of \Cref{lem:reduc_w}}
\label{app:long shot}

\begin{proof}[Proof of \Cref{lem:reduc_w}]
    From \Cref{eq:recursionwi}, it holds that $\Wn(\lambda, X_{2:n}) = \Wn(X_{2:n}) = \lambda$.

    Then, again from \Cref{eq:recursionwi} and using the notation $(a)_+ = \max(a, 0)$ we have
    \begin{align*}
        \Wn(X_{1:n}) 
        &= \max_{Q_1 \in \mech_\varepsilon} \EE[\max(\Wn(X_{2:n}), \EE[X_1 | Z_1])] \\
        &= \max_{Q_1 \in \mech_\varepsilon} \EE[\Wn(X_{2:n}) + \big(\EE[X_1 | Z_1] - \Wn(X_{2:n}) \big)_+] \\
        &= \underbrace{\Wn(X_{2:n})}_{=\Wn(\lambda, X_{2:n})} + \max_{Q_1 \in \mech_\varepsilon} \EE[\big(\EE[X_1 | Z_1] - \underbrace{\Wn(X_{2:n})}_{=\lambda} \big)_+].
    \end{align*}
    We therefore obtain
    \begin{equation}
    \label{eq:W_1 lambda}
        \Wn(X_{1:n}) = \Wn(\lambda, X_{2:n}) + \max_{Q_1 \in \mech_\varepsilon} \EE[\big(\EE[X_1 | Z_1] - \lambda \big)_+].
    \end{equation}

    Let us now upper-bound the value of the prophet
    \begin{align*}
        \Ln(X_{1:n}) &\leq \Ln(\lambda, X_{1:n}) \\
        &= \max_{Q_{1}, \dots, Q_n \in \mech_\varepsilon} \EE[\max(\lambda, \EE[X_1 | Z_1], \dots, \EE[X_n | Z_n])] \\
        &= \max_{Q_{1}, \dots, Q_n \in \mech_\varepsilon} \EE[\max_{i \in \{2, \dots, n\}}(\lambda, \EE[X_i | Z_i])  \\ &~~+ \big(\EE[X_1 | Z_1] - \underbrace{\max_{i \in \{2, \dots, n\}}(\lambda, \EE[X_i | Z_i])}_{\geq \lambda}) \big)_+] \\
        &\leq \max_{Q_{2:n} \in \mech_\varepsilon} \EE[\max_{i \in \{2, \dots, n\}}(\lambda, \EE[X_i | Z_i])] + \max_{Q_1 \in \mech_\varepsilon}\EE[\big(\EE[X_1 | Z_1] - \lambda) \big)_+] \\
         &= \Ln(\lambda, X_{2:n}) + \max_{Q_1 \in \mech_\varepsilon}\EE[\big(\EE[X_1 | Z_1] - \lambda) \big)_+] \\
    \end{align*}
    We have shown
    \begin{equation}
        \label{eq:ldp lambda}
        \Ln(X_{1:n}) \leq \Ln(\lambda, X_{2:n}) + \max_{Q_1 \in \mech_\varepsilon}\EE[\big(\EE[X_1 | Z_1] - \lambda) \big)_+]. \\
    \end{equation}
    From \Cref{eq:W_1 lambda,eq:ldp lambda}, we get
    \begin{align*}
    \frac{\Wn(X_{1:n})}{\Ln(X_{1:n})} &\geq \frac{\Wn(\lambda, X_{2:n}) + \max_{Q_1 \in \mech_\varepsilon} \EE_{\substack{Z_1 \sim Q_1(\cdot | X_1) \\ X_1 \sim \Dcal_1}}[\big(\EE[X_1 | Z_1] - \lambda \big)_+]}{\Ln(\lambda, X_{2:n}) + \max_{Q_1 \in \mech_\varepsilon} \EE_{\substack{Z_1 \sim Q_1(\cdot | X_1) \\ X_1 \sim \Dcal_1}}[\big(\EE[X_1 | Z_1] - \lambda \big)_+]} \\
    &\geq \frac{\Wn(\lambda, X_{2:n})}{\Ln(\lambda, X_{2:n})}
    \end{align*}
    which concludes the proof.
\end{proof}

\subsubsection{Proof of \Cref{lem:long shot}}
\label{app:reduc_w}

\begin{proof}[Proof of \Cref{lem:long shot}]
    From \Cref{eq:recursionwi}, we have on one hand
    \begin{align*}
        \Wn(X_{n-2}, X_{n-1}, X_n) 
        &= \max_{Q_{n-2} \in \mech_\varepsilon}\EE[\max(\Wn(X_{n-1}, X_{n}), \EE[X_{n-2} | Z_{n-2}])] \\
        &= \Wn(X_{n-2}, \Wn(X_{n-1}, X_n)).
    \end{align*}
    On the other hand, fix $\alpha > 0$, we can write
    \begin{align*}
        \Wn(X_{n-2}, L_\alpha) &= \EE_{Q_{n-2} \in \mech_\varepsilon}[\max(\Wn(L_\alpha), \EE[X_{n-2} | Z_{n-2}]]\\
        &= \Wn(X_{n-2}, \Wn(X_{n-1}, X_n))
    \end{align*}
     where the last inequality follows since $\Wn(L_\alpha) = \EE[L_\alpha] = \Wn(X_{n-1}, X_n)$. 
     Then by recurrence, we show that for any $i \in [n-2]$, $\Wn(X_{i:n}) = \Wn(X_{i:n-2}, L_\alpha)$. The case $i=n-2$ is done above, and the induction follows from \Cref{eq:recursionwi} which implies $\Wn(X_i, X_{i+1:n}) = \Wn(X_i, \Wn(X_{i+1:n}))$. 
     
     Taking, $i=1$, we obtain
     \begin{equation}
     \label{eq:long shot W}
         \forall \alpha > 0, \Wn(\lambda, X_{2:n}) = \Wn(\lambda, X_{2:n-2}, L_\alpha).
     \end{equation}

     We now show that $\Ln(\lambda, X_{2:n})(1 - \alpha) \leq \Ln(\lambda, X_{2:n-2}, L_\alpha)$ which with \Cref{eq:long shot W} will conclude the proof.

     First, fix $\alpha < \frac12$ and  let us lower bound
     \begin{align*}
         &\Ln(\lambda, X_{2:n-2}, L_{\alpha}) \\
         &= \max_{Q_{2:{n-1}} \in \mech_\varepsilon} \EE[\max \bigg(\underbrace{\max_{i=2}^{n-2}\big(\lambda, \EE[X_i | Z_i]\big)}_{:=A}, \EE[L_\alpha | Z_{n-1}] \bigg)].
     \end{align*}
     
     Let us pick $Q_{n-1} = Q \in \mech_\varepsilon$ where
     \[
     Q: \ell \in \RR_+ \mapsto Q(Z | \ell) \in \Delta(\{0, 1\})
     \]
    is defined by
     \begin{align*}
         &\forall \ell > 0, Q(Z = 1 | \ell) = \frac{\exp(\varepsilon)}{\exp(\varepsilon) + 1} := p \\
         &\forall \ell > 0, Q(Z = 0 | \ell) = 1 - p \\
         &Q(Z = 1 | \ell=0) = 1 - p \\
         &Q(Z = 0 | \ell=0) = p.
     \end{align*}

     Call $Z$ the $\varepsilon$-LDP view of $L_\alpha$ through $Q$ (meaning $Z | L \sim Q(\cdot | L))$). We have
     \begin{align*}
         &\Ln(\lambda, X_{2:n-2}, L_{\alpha}) \\
         &\geq \max_{Q_{2:{n-2}} \in \mech_\varepsilon} \EE[\max (A, \EE[L_\alpha | Z])] \\
         &= \max_{Q_{2:{n-2}}} \sum_{i=0}^1 \PP(Z = i)\EE[\max (A, \EE[L_\alpha | Z=i])] \\
         &= \max_{Q_{2:{n-2}}} \sum_{i=0}^1 \EE[\max (\PP(Z = i) A, \EE[L_\alpha \ind{Z=i}])] \\
         &\overset{(1)}{=} \max_{Q_{2:{n-2}}} \EE\bigg[\max\big(A, \EE[L_\alpha], \PP(Z=0) A +  \EE[L_\alpha \mathds{1}[Z=1]], \\&~~\PP(Z=1) A +  \EE[L_\alpha \mathds{1}[Z=0]]\big) \bigg] \\
         &\overset{(2)}{=} \max_{Q_{2:{n-2}}} \EE\bigg[\max\big(A, \PP(Z=0) A +  \EE[L_\alpha \mathds{1}[Z=1]], \\&~~\PP(Z=1) A +  \EE[L_\alpha \mathds{1}[Z=0]]\big) \bigg]
     \end{align*}
     where $(1)$ is by $\max(a, b) + \max(c, d) = \max(a + c, a+d, b+c, b+d)$ and $(2)$ follows from $A \geq \Wn(X_{n-1}, X_n) = \EE[L_\alpha]$.

     From the definition of $Q$ and $L_\alpha$, we get
     \begin{align*}
      &\PP(Z = 1) = p \alpha + (1 - p)(1 - \alpha) \\
      &\PP(Z = 0) = (1 - p) \alpha + p(1 - \alpha) \\
      &\EE[L_\alpha \mathds{1}[Z=1]] = \Wn(X_{n-1}, X_n) p \\
      &\EE[L_\alpha \mathds{1}[Z=0]] = \Wn(X_{n-1}, X_n) (1 - p).
     \end{align*}
     Notice that $\PP(Z = 1) < \PP(Z=0)$ and $\EE[L_\alpha \mathds{1}[Z=0]] < \EE[L_\alpha \mathds{1}[Z=0]] $ since $\alpha < \frac12$.

     After substitution, we get:
     \begin{align*}
         &\Ln(\lambda, X_{2:n-2}, L_{\alpha}) \\
         &\geq \max_{Q_{2:{n-2}}} \EE\bigg[\max\big(A, \PP(Z=0) A +  \EE[L_\alpha \mathds{1}[Z=1]] \big) \bigg] \\
         &= \max_{Q_{2:{n-2}}} \EE\bigg[\max\big(A, (p - \alpha(2p - 1)) A +  p \Wn(X_{n-1}, X_n) \big) \bigg] \\
         &\geq \max_{Q_{2:{n-2}}} \EE\bigg[\max\big(A - \alpha A, pA - \alpha A +  p \Wn(X_{n-1}, X_n) \big) \bigg] \\
         &\geq \max_{Q_{2:{n-2}}} \EE\bigg[\max\big(A, p (A +  \Wn(X_{n-1}, X_n)) \big) \bigg] - \max_{Q_{2:{n-2}}} \alpha \EE[A] \\
         &\geq \max_{Q_{2:{n-2}}} \EE\bigg[\max\big(A, p (A +  \Wn(X_{n-1}, X_n)) \big) \bigg] - \alpha \Ln(\lambda, X_{2:n})
     \end{align*}

     To conclude the proof, we need to show that
     \[
     \Ln(\lambda, X_{2:n-2}, X_{n-1:n}) \leq \EE[\max(A, p(A + \Wn(X_{n-1}, X_n)))].
     \]

     We have
     \begin{align*}
         &\Ln(\lambda, X_{2:n-2}, X_{n-1:n}) \\
         &= \max_{Q_{2:n} \in \mech_\varepsilon} \EE[\max \bigg(A, \underbrace{\EE[X_{n-1} \mid Z_{n-1}]}_{:=Y_{n-1}}, \underbrace{\EE[X_{n} \mid Z_{n}]}_{:=Y_n}\bigg)]
     \end{align*}
     Fix $Q_{2:n-1}$, call $\pi_A, \pi_{Y_n}, \pi_{Y_{n-1}}$ the distribution of $A, Y_n, Y_{n-1}$ respectively and denote $p = \exp(\varepsilon) / (1 + \exp(\varepsilon))$. We have
     \begin{align*}
         &\max_{Q_n} \EE[\max (A, Y_{n-1}, Y_n)] \\
         &\overset{(1)}{=} \max_{Q_n} \int \EE[\max(a, y, Y_n )] d \pi_A(a) d \pi_{Y_{n-1}}(y) \\
         &\overset{(2)}{\leq} \int \max_{Q_n} \EE[\max(a, y, Y_n )] d\pi_A(a) d \pi_{Y_{n-1}}(y)\\
         &\overset{(3)}{=} \int \max\Big(a, y, \EE[X_n], \EE[\phi(X_n, \max(a, y))]\Big) d\pi_A(a) d \pi_{Y_{n-1}}(y) \\
         &\overset{(4)}{=} \int \max\Big(a, y, \EE[\phi(X_n, \max(a, y))]\Big) d\pi_A(a) d \pi_{Y_{n-1}}(y) \\
         &\overset{(5)}{\leq}\int \max\Big(a, y, \EE[p (X_n + \max(a, y))]\Big) d\pi_A(a) d \pi_{Y_{n-1}}(y) \\
         &\overset{(6)}{=} \int (p \max(a, y) + \max((1-p) \max(a, y), p \EE[X_n])) d\pi_A(a) d \pi_{Y_{n-1}}(y) \\
         &= \int \EE\Big[ p \max(a, Y_{n-1}) + (1 - p)\max(Y_{n-1}, \max(a, \frac{p}{1 - p} \EE[X_n])) \Big]d \pi_{A}(a)
     \end{align*}
     where $(1)$ is by mutual independence of $A, Y_n, Y_{n-1}$, $(2)$ is by Jensen's inequality, $(3)$ is by \Cref{prop:optimalchannel}, $(4)$ is by $A \geq \EE[X_n]$, $(5)$ is by $\phi(x, y) \leq p(x+ y)$, $(6)$ is by $\max(a + c, b + c) = \max(a, b) + c$. 
     
     Taking the maximum with respect to $Q_{n-1}$ on both sides, we get
     \begin{align*}
         &\max_{Q_{n-1:n}} \EE[\max(A, Y_{n-1}, Y_n)] \\ 
         &\leq \max_{Q_{n-1}} \int \EE\Big[ p \max(a, Y_{n-1}) + (1 - p)\max(Y_{n-1}, \max(a, \frac{p}{1 - p} \EE[X_n]))\Big]d \pi_{A}(a) \\
         &\overset{(1)}{\leq} \int \max_{Q_{n-1}} \EE\Big[ p \max(a, Y_{n-1}) + (1 - p) \max(Y_{n-1}, \max(a, \frac{p}{1 - p} \EE[X_n]))\Big]d \pi_{A}(a) \\
         &\overset{(2)}{\leq} \int \bigg(p \max_{Q_{n-1}} \EE\Big[ \max(a, Y_{n-1})\Big] \\& + (1 - p)\max_{Q_{n-1}} \EE\Big[\max(Y_{n-1}, \max(a, \frac{p}{1 - p} \EE[X_n]))\Big]\bigg) d\pi_{A}(a) \\
         &\overset{(3)}{=} \int \bigg( p \max(a, \EE[X_{n-1}], \EE[\phi(X_{n-1}, a)]) \\& + (1 - p) \max(a, \frac{p}{1 - p} \EE[X_n], \EE[X_{n-1}], \EE[\phi(X_{n-1}, \max(a, \frac{p}{1 - p} \EE[X_n]))]) \bigg) d\pi_{A}(a) \\
         &\overset{(4)}{\leq} \int \bigg( p \max(a, \EE[\phi(X_{n-1}, a)]) \\& + (1 - p) \max(a, \frac{p}{1 - p} \EE[X_n], \EE[\phi(X_{n-1}, \max(a, \frac{p}{1 - p} \EE[X_n]))]) \bigg) d\pi_{A}(a) \\
         &= \int \bigg( \ind{a \geq \frac{p}{1 - p} \EE[X_n]} \bigg[ p \max(a, \EE[\phi(X_{n-1}, a)]) + (1 - p) \max(a, \EE[\phi(X_{n-1}, a)] \bigg] + \\
         &~\ind{a < \frac{p}{1 - p} \EE[X_n]} \bigg[ p \max(a, \EE[\phi(X_{n-1}, a)]) \\& + (1 - p) \max(\frac{p}{1 - p} \EE[X_n], \EE[\phi(X_{n-1}, \frac{p}{1 - p} \EE[X_n])]) \bigg] \bigg) d\pi_{A}(a) \\
         &\overset{(5)}{=} \int \bigg( \ind{a \geq \frac{p}{1 - p} \EE[X_n]} \bigg[ \max(a, \EE[\phi(X_{n-1}, a)]) \bigg] + \\
         &~\ind{a < \frac{p}{1 - p} \EE[X_n]} \bigg[ p \max(a, \EE[\phi(X_{n-1}, a)]) \\& + \max(p \EE[X_n], \EE[\phi((1 - p) X_{n-1}, p \EE[X_n])]) \bigg] \bigg) d\pi_{A}(a) \\
         &\overset{(6)}{\leq} \int \bigg(  \max(a, \EE[\phi(X_{n-1}, a)], \\
         &~ \underbrace{p \max(a, \EE[\phi(X_{n-1}, a)]) + \max(p \EE[X_n], \EE[\phi((1 - p) X_{n-1}, p \EE[X_n])])}_{(i)}) \bigg)d\pi_{A}(a) \\
         &\overset{(7)}{\leq} \int \bigg(  \max(a, p(a + \EE[X_{n-1}]), (i) ) \bigg)d\pi_{A}(a) \\
         &\overset{(8)}{\leq} \int \bigg(  \max(a, p(a + \Wn(X_{n-1}, X_n)), (i) ) \bigg)d\pi_{A}(a)
     \end{align*}
     where $(1)$ is from Jensen inequality, $(2)$ is from \\
     $\max_Q (f(Q) + g(Q)) \leq \max_Q f(Q) + \max_Q g(Q)$, $(3)$ is from \Cref{prop:optimalchannel}, $(4)$ is by $A \geq \EE[X_n]$ and $A \geq \EE[X_{n-1}]$, $(5)$ is by $c \phi(x, y) = \phi(cx, cy)$, $(6)$ is by $\ind{E} x + \ind{\overline{E}} y \leq \max(x, y)$ where $E$ and $\overline{E}$ are complementary event, $(7)$ is by $\phi(x, y) \leq p(x + y)$ and $(8)$ is by $\EE[X_{n-1}] \leq \Wn(X_{n-1}, X_n)$.

     We will now show that $(i) \leq p(a + \Wn(X_{n-1}, X_n))$. We have 
     \[
     (i) = \max( (a), (b))
     \]
     where
     \begin{align*}
         (a) &= p (\max(a, \EE[\phi(X_{n-1}, a)]) + \EE[X_n]) \\
         (b) &= p \max(a, \EE[\phi(X_{n-1}, a)]) + \EE[\phi((1 - p) X_{n-1}, p \EE[X_n])] \\
         &= \max\big(\underbrace{pa + \EE[\phi((1 - p) X_{n-1}, p \EE[X_n])]}_{(b1)},\\
         &\underbrace{\EE[p \phi(X_{n-1}, a) + \phi((1 - p) X_{n-1}, p \EE[X_n])]}_{(b2)} \big).
     \end{align*}

     First, we have
     \begin{align*}
         (a) &= p (\max(a, \EE[\phi(X_{n-1}, a)]) + \EE[X_n]) \\
         &= p( a + \max(0, \EE[\phi(X_{n-1}, a)] - a) + \EE[X_n]) \\
         &= p( a + \max(0, \EE[\phi(X_{n-1} -a, 0)] ) + \EE[X_n]) && \text{($\phi(x, y) -c =\phi(x - c, y -c) $)} \\
         &\leq p( a + \max(0, \EE[\phi(X_{n-1} -\EE[X_n], 0)] ) + \EE[X_n]) && \text{($\EE[X_n] \leq a$)} \\
         &= p( a + \max(\EE[X_n], \EE[\phi(X_{n-1}, \EE[X_n])])) \\
         &\leq p(a + \Wn(X_{n-1}, X_{n})).
     \end{align*}

     Then, it holds
     \begin{align*}
         (b1) &= pa + \EE[\phi((1 - p) X_{n-1}, p \EE[X_n])] \\
         &\leq p(a + (1 - p) \EE[X_{n-1}] + p\EE[X_n]) && \text{(by $\phi(x,y) \leq p(x + y)$)} \\
         &\leq p(a + \Wn(X_{n-1}, X_n))
     \end{align*}
     where the last inequality follows from $\EE[X_{n-1}] \leq \Wn(X_{n-1}, X_n)$ and $\EE[X_{n}] \leq \Wn(X_{n-1}, X_n)$.
     
     Lastly, noticing that
     \begin{equation}
     \phi(x, y) = \phi(y, x) = px + (1 - p) y + (2p -1) (y - x)_+,
     \label{eq:phi diff}
     \end{equation}
     we write
     \begin{align*}
         (b2) &= \EE[p \phi(X_{n-1}, a) + \phi((1 - p) X_{n-1}, p \EE[X_n])] \\
         &\overset{(1)}{\leq} \EE[p \phi(X_{n-1}, a)] + p((1-p)\EE[X_{n-1}] + p \EE[X_n]) \\
         &\overset{(2)}{=} p(pa + (1 - p) \EE[X_{n-1}] + (2p - 1) \EE[(X_{n-1} - a)_+]) \\ 
         &~~ + p((1-p)\EE[X_{n-1}] + p \EE[X_n]) \\
         &= p(p \EE[X_n] + (1 - p) \EE[X_{n-1}] + (2p-1) \EE[(X_{n-1} - a)_+]) \\
         &~~+ p((1 - p) \EE[X_{n-1}] + pa) \\
         &\overset{(3)}{\leq} p(p \EE[X_n] + (1 - p) \EE[X_{n-1}] + (2p-1) \EE[(X_{n-1} - \EE[X_n])_+]) + pa \\
         &\overset{(4)}{=} p(\EE[\phi(X_{n-1}, \EE[X_n])]) + pa \\
         &= p(a + W(X_{n-1}, X_n)),
     \end{align*}
     where $(1)$ is from $\phi(x, y) \leq p(x + y)$, $(2)$ is by \Cref{eq:phi diff}, $(3)$ uses $a \geq \EE[X_n]$ and $a \geq \EE[X_{n-1}]$ and $(4)$ is by \Cref{eq:phi diff}.

     We have therefore shown that
     \[
     (i) \leq p(a + \Wn(X_{n-1}, X_n))
     \]
     and hence we have
     \begin{align*}
         &\max_{Q_{n-1:n}} \EE[\max(A, Y_{n-1}, Y_n)] \\
         &\leq \int \bigg(  \max(a, p(a + \Wn(X_{n-1}, X_n)), (i) ) \bigg)d\pi_{A} \\
         &= \int \bigg(  \max(a, p(a + \Wn(X_{n-1}, X_n))) \bigg)d\pi_{A} \\
         &= \EE[\max(A, p(A + \Wn(X_{n-1}, X_n)))].
     \end{align*}
     which concludes the proof.
\end{proof}

\end{document}